\documentclass[sigplan,10pt,dvipsnames,screen]{acmart}
\usepackage{stmaryrd}
\usepackage{xspace}
\usepackage{mathpartir}
\usepackage{scalerel}
\usepackage{wasysym}
\usepackage{enumitem}

\copyrightyear{2022}
\setcopyright{acmlicensed}
\acmPrice{15.00}
\acmDOI{10.1145/3519939.3523703}
\acmYear{2022}
\acmSubmissionID{pldi22main-p300-p}
\acmISBN{978-1-4503-9265-5/22/06}
\acmConference[PLDI '22]{Proceedings of the 43rd ACM SIGPLAN International Conference on Programming Language Design and Implementation}{June 13--17, 2022}{San Diego, CA, USA}
\acmBooktitle{Proceedings of the 43rd ACM SIGPLAN International Conference on Programming Language Design and Implementation (PLDI '22), June 13--17, 2022, San Diego, CA, USA}
\startPage{1}

\usepackage{booktabs}   %% For formal tables:
\usepackage{subcaption} %% For complex figures with subfigures/subcaptions
\usepackage{bm}
\usepackage{float}
\begin{document}

%% Title information
\title{Semantic Soundness for Language Interoperability}         %% [Short Title] is optional;
                                        %% when present, will be used in
                                        %% header instead of Full Title.
% \titlenote{with title note}             %% \titlenote is optional;
                                        %% can be repeated if necessary;
                                        %% contents suppressed with 'anonymous'
% \subtitle{Subtitle}                     %% \subtitle is optional
% \subtitlenote{with subtitle note}       %% \subtitlenote is optional;
                                        %% can be repeated if necessary;
                                        %% contents suppressed with 'anonymous'

%% Author information
%% Contents and number of authors suppressed with 'anonymous'.
%% Each author should be introduced by \author, followed by
%% \authornote (optional), \orcid (optional), \affiliation, and
%% \email.
%% An author may have multiple affiliations and/or emails; repeat the
%% appropriate command.
%% Many elements are not rendered, but should be provided for metadata
%% extraction tools.

\author{Daniel Patterson}
\affiliation{
  \institution{Northeastern University}            %% \institution is required
  \streetaddress{440 Huntington Avenue}
  \city{Boston}
  \state{MA}
  \postcode{02115}
  \country{USA}
}
\email{dbp@dbpmail.net}          %% \email is recommended

\author{Noble Mushtak}
\affiliation{
  \institution{Northeastern University}            %% \institution is required
  \streetaddress{440 Huntington Avenue}
  \city{Boston}
  \state{MA}
  \postcode{02115}
  \country{USA}
}
\email{mushtak.n@northeastern.edu}          %% \email is recommended

\author{Andrew Wagner}
\affiliation{
  \institution{Northeastern University}            %% \institution is required
  \streetaddress{440 Huntington Avenue}
  \city{Boston}
  \state{MA}
  \postcode{02115}
  \country{USA}
}
\email{ahwagner@ccs.neu.edu}          %% \email is recommended

\author{Amal Ahmed}
\affiliation{
  \institution{Northeastern University}            %% \institution is required
  \streetaddress{440 Huntington Avenue}
  \city{Boston}
  \state{MA}
  \postcode{02115}
  \country{USA}
}
\email{amal@ccs.neu.edu}          %% \email is recommended

\makeatletter
\g@addto@macro \small {%
  \setlength\abovedisplayskip{1pt}%
  \setlength\belowdisplayskip{1pt}%
  \setlength\abovedisplayshortskip{1pt}%
  \setlength\belowdisplayshortskip{1pt}%
  \setlength\abovecaptionskip{10pt}%
  \setlength\belowcaptionskip{2pt}%
  \setlength{\textfloatsep}{5pt}
  \setlength{\floatsep}{5pt}
}
\makeatother

% NOTE(dbp 2021-07-08): Make mathpartir a little more compact
\mprset {sep=1.2em}

\newcommand{\dbp}[1]{{\color{Red}(DBP: #1)}}
\newcommand{\aja}[1]{{\color{Red}(AA: #1)}}
\newcommand{\noble}[1]{{\color{Red}(NM: #1)}}
\newcommand{\ahw}[1]{{\color{Red}(AHW: #1)}}

\newcommand{\para}[1]{\vspace{0.1cm}\emph{#1}}

\newcommand{\vo}{\mathrel{\color{black}{\vert}}}
\newcommand{\mtt}[1]{\mathtt{#1}}
\newcommand{\mbt}[1]{\mathbf{\mathtt{#1}}}

\newcommand{\close}{\text{close}}
\newcommand{\closefirst}{\text{closeaffine}}

\newcommand{\gcmov}{\tlang{gcmov}}
\newcommand{\thnk}{\text{thunk}}
\newcommand{\grd}{\text{guard}}
\newcommand{\prtct}{\text{protect}}
\newcommand{\fl}{\mathit{FL}}
\newcommand\mpart[1]{{#1} : MHeap}
\newcommand\gpart[1]{{#1} : GCHeap}
\newcommand{\rchgclocs}{\text{rchgclocs}}
\newcommand{\rl}{\text{reachablelocs}}
\newcommand{\rf}{\text{flags}}
\newcommand{\ff}{\text{freeflags}}
\newcommand{\fv}{\text{freevars}}
\newcommand{\cod}{\text{cod}}

\newcommand{\callgc}{\tlang{callgc}}

\newcommand\affwf[2]{\tlang{#1} \vdash^A \tlang{#2}}
\newcommand\afffn[1]{\vdash^\lambda \tlang{#1}}

\newcommand{\aused}{\textsc{used}}
\newcommand{\aunused}{\textsc{unused}}

\newcommand\scale[2]{\vstretch{#1}{\hstretch{#1}{#2}}}
\newcommand\envplus{\ensurestackMath{\stackinset{c}{0ex}{c}{0.2ex}{\scale{.6}{\scriptscriptstyle +}}{\Cup}}}

\newcommand{\W}{\mathit{W}}
\newcommand{\heapty}{\Psi}
\newcommand{\codety}{\raisebox{2pt}{$\chi$}}
\newcommand{\bij}{\mathtt{bij}}
\newcommand{\byte}{\mathtt{byte}}
\newcommand{\loc}{\ell}
\newcommand{\flag}{\mathit{f}}
\newcommand{\Vrel}[2]{\mathcal{V}\llbracket #1 \rrbracket_{#2}}
\newcommand{\HVrel}[2]{\mathcal{HV}\llbracket #1 \rrbracket_{#2}}
\newcommand{\CVrel}[2]{\mathcal{CV}\llbracket #1 \rrbracket_{#2}}
\newcommand{\Erel}[2]{\mathcal{E}\llbracket #1 \rrbracket_{#2}}
\newcommand{\Krel}[2]{\mathcal{K}\llbracket #1 \rrbracket_{#2}}
\newcommand{\Rrel}[2]{\mathcal{R}\llbracket #1 \rrbracket_{#2}}
\newcommand{\obs}{\mathcal{O}}
\newcommand{\Envrel}[2]{\mathcal{G}\llbracket #1 \rrbracket_{#2}}
\newcommand{\Tyenvrel}[1]{\mathcal{D}\llbracket #1 \rrbracket}
\newcommand{\worldext}{\sqsubseteq}
\newcommand{\worldextstrict}{\sqsubset}
\newcommand{\floor}[2]{\lfloor #2 \rfloor_{#1}}
\newcommand{\dom}[1]{\text{dom}(#1)}
\newcommand{\heap}{\mathsf{H}}
\newcommand{\locsets}{\mathbb{L}}
\newcommand{\locsett}{\mathit{L}}
\newcommand{\locbij}{\mathit{\eta}}
\newcommand{\codeheap}{\mathsf{C}}
\newcommand{\stack}{\gamma}
\newcommand{\nostack}{\emptyset}
\newcommand{\astore}{\Theta}
\newcommand{\fstore}{\Phi}
\newcommand{\cont}{\mathsf{K}}
\newcommand{\Stack}{\Theta}
\newcommand{\steps}[1]{\overset{#1}{\rightarrow}}
\newcommand{\irred}[1]{\text{irred}(#1)}
\newcommand{\err}[1]{\text{err}(#1)}
\newcommand{\closure}[1]{\langle #1 \rangle}
\newcommand{\qual}{\mathfrak{q}}
\newcommand{\unrestricted}{\tlang{U}}
\newcommand{\affine}{\tlang{A}}

\newcommand{\sem}[1]{\llbracket #1 \rrbracket}

\newcommand{\rcolor}[1]{{\color{RoyalBlue} #1}}
\newcommand{\rlang}[1]{\mbox{\unboldmath$\mathtt{\rcolor{#1}}$}}
\newcommand{\langref}{\rlang{MiniML}\xspace}
\newcommand{\ccolor}[1]{{\color{Magenta} #1}}
\newcommand{\clang}[1]{\mathtt{\ccolor{#1}}}
\newcommand{\langcast}{\clang{Cast}\xspace}
\newcommand{\tcolor}[1]{{\color{Black} #1}}
\newcommand{\tlang}[1]{\mathsf{\tcolor{#1}}}
% not at runtime
\newcommand{\gcolor}[1]{{\color{Gray} #1}}
\newcommand{\glang}[1]{\mathsf{\gcolor{#1}}}

\newcommand{\langtarget}{\tlang{LCVM}\xspace}
\newcommand{\langlltarget}{\tlang{Mem}\xspace}
\newcommand{\acolor}[1]{{\color{orange}{#1}}}
\newcommand{\alang}[1]{\bm{\mathrm{\acolor{#1}}}}
\newcommand{\langatlc}{\alang{\textsc{\textbf{Affi}}}\xspace}
\newcommand{\Affenv}{\Omega}

\newcommand{\rhlcolor}[1]{{\color{Turquoise} #1}}
\newcommand{\rhlang}[1]{\mathtt{\rhlcolor{#1}}}
\newcommand{\langrhl}{\rhlang{RefHL}\xspace}
\newcommand{\rllcolor}[1]{{\color{BurntOrange} #1}}
\newcommand{\rllang}[1]{\bm{\mathrm{\rllcolor{#1}}}}
\newcommand{\langrll}{\rllang{RefLL}\xspace}

\newcommand{\langstack}{\tlang{StackLang}\xspace}

\newcommand{\taub}{\bm\tau}

\newcommand{\tconf}[2]{\langle \tcolor{#1},\tcolor{#2} \rangle}
\newcommand{\sconf}[3]{\langle \tcolor{#1};\tcolor{#2};\tcolor{#3} \rangle}
\newcommand{\phconf}[3]{\langle \tcolor{#1},\tcolor{#2},\tcolor{#3} \rangle}

\newcommand{\mconf}[5]{\langle \tcolor{#1},\tcolor{#2},\tcolor{#3},\tcolor{#4},\tcolor{#5} \rangle}

\newcommand{\step}{\rightarrow}
\newcommand{\multistep}{\overset{*}{\step}}
\newcommand{\progstep}{\rightarrow}
\newcommand{\primstep}{\Mapsto}
\newcommand{\kontstep}{\overset{ctx}{\step}}
\newcommand{\gcstep}{\overset{gc}{\step}}
\newcommand{\multigcstep}{\overset{*}{\gcstep}}
\newcommand{\kontstuck}{\overset{ctx}{\nrightarrow}}

\newcommand{\phstep}{\dashrightarrow}
\newcommand{\phsteps}[1]{\overset{#1}{\phstep}}
\newcommand{\phmultistep}{\overset{*}{\phstep}}

\newcommand{\done}{\ddagger}
\newcommand{\affref}{\mathfrak{a}}
\newcommand{\used}{\boxast}

\newcommand{\restore}[3]{restore~{#1}{\mapsto}{#2}~{#3}}

\newcommand{\convsub}{\leq}
\newcommand{\conveq}{\cong}
\newcommand{\conv}{\sim}

\newcommand{\lconvert}{\llparenthesis}
\newcommand{\rconvert}{\rrparenthesis}
\newcommand{\convert}[2]{\lconvert #2 \rconvert_{#1}}

\newcommand{\aprx}{\preceq}

\newcommand{\context}{\mathfrak{C}}
\newcommand{\produces}{\,{\rightsquigarrow}\,}

\newcommand{\term}{\Downarrow}
\newcommand{\running}{\text{running}}

\newcommand{\later}{{\vartriangleright}}

\newcommand{\dyn}{\circ}
\newcommand{\stat}{\bullet}
\newcommand{\modal}{{\scaleobj{0.55}{{\LEFTcircle}}}}
\newcommand{\modalolli}{\rlap{\text{---}}\hspace{0.5em}\raisebox{0.12em}{\modal~}}
\newcommand{\statlolli}{\multimapdot}
\newcommand{\dynlolli}{\multimap}

% First Case Study Defs
% NOTE(dbp 2021-10-20): If the following looks odd, it's because we added
% mutable references to `System F', which made it identical to `MiniML'.
\newcommand{\fcolor}[1]{\rcolor{#1}}
\newcommand{\flang}[1]{\rlang{#1}}
\newcommand{\langsysf}{\langref}
\newcommand{\lcolor}[1]{{\color{Magenta}{#1}}}
\newcommand{\llang}[1]{\bm{\mathrm{\lcolor{#1}}}}
\newcommand{\langlt}{\llang{\textsc{\textbf{L\textsuperscript{3}}}}\xspace}
\newcommand{\foreign}[1]{\left\langle #1 \right\rangle}
\newcommand{\foreignset}{\textsc{Duplicable}}
\newcommand{\Ptr}[1]{ptr\,#1}
\newcommand{\Capa}[2]{cap\,#1\,#2}
\newcommand{\ptr}[1]{ptr\,#1}
\newcommand{\capa}{cap}
\newcommand{\pack}[1]{\left\ulcorner #1 \right\urcorner}
\newcommand{\atf}[1]{#1 . \flang{F}}
\newcommand{\atl}[1]{#1 . \llang{L3}}
\newcommand{\letc}[3]{let~#1 = #2~in~#3}

\newcommand{\sstack}{\mathsf{S}}

\definecolor{hlgray}{gray}{0.9}
\newcommand{\hl}[1]{\colorbox{hlgray}{#1}}

% Second Case Study Defs
\newcommand{\lts}{\gamma_{\text{locs}}} % locs-to-subst
\newcommand{\fsub}{\gamma_{\fcolor{\Gamma}}}
\newcommand{\lsub}{\gamma_{\llang{L}}}
\newcommand{\lgammasub}{\gamma_{\llang{L}}.\llang{\Gamma}}
\newcommand{\ldeltasub}{\gamma_{\llang{L}}.\llang{\Delta}}

% error constants
\newcommand{\fail}[1]{\tlang{fail~#1}}
\newcommand{\tyerrcode}{\textsc{Type}}
\newcommand{\tyfail}{\fail{\tyerrcode}}
\newcommand{\converrcode}{\textsc{Conv}}
\newcommand{\convfail}{\fail{\converrcode}}
\newcommand{\idxerrcode}{\textsc{Idx}}
\newcommand{\idxfail}{\fail{\idxerrcode}}
\newcommand{\refokerr}{\textsc{OkErr}}
\newcommand{\inlokerr}{\{\converrcode, \idxerrcode\}}
\newcommand{\ptrerrcode}{\textsc{Ptr}}
\newcommand{\ptrfail}{\fail{\ptrerrcode}}

\newcommand{\Typ}{\ensuremath{\mathit{Typ}}}
\newcommand{\UnrTyp}{\ensuremath{\mathit{UnrTyp}}}

\newcommand{\gcmaps}{\overset{gc}{\mapsto}}
\newcommand{\manmaps}{\overset{m}{\mapsto}}

% refs:
% https://tex.stackexchange.com/a/85506
% https://tex.stackexchange.com/a/140404
\newcommand{\sqSubset}{~\hbox{$\sqsubset$}\kern -4.5pt\hbox{$\mathbin{\vcenter{\hbox{\scaleobj{0.50}{\bullet}}}}$~}}
\newcommand{\sqSubseteq}{~\hbox{$\sqsubseteq$}\kern -4pt\raisebox{0.5pt}{$\mathbin{\vcenter{\hbox{\scaleobj{0.50}{\bullet}}}}$~}}
\newcommand{\worldextstricte}{\sqSubset}
\newcommand{\worldexte}{\sqSubseteq}

%%% Local Variables:
%%% mode: latex
%%% TeX-master: "paper"
%%% End:

%% Abstract
%% Note: \begin{abstract}...\end{abstract} environment must come
%% before \maketitle command
\begin{abstract}
  Programs are rarely implemented in a single language, and thus questions of type
soundness should address not only the semantics of a single language, but how it
interacts with others. Even between type-safe languages, disparate features can
frustrate interoperability, as invariants from one language can easily be
violated in the other. In their seminal 2007 paper, \citet{matthews07} proposed
a multi-language construction that augments the interoperating languages with
a pair of \emph{boundaries} that allow code from one language to be embedded in
the other. While this technique has been widely applied, their syntactic source-level
interoperability doesn't reflect practical implementations, where the behavior
of interaction is only defined after compilation to a common target, and any
safety must be ensured by target invariants or inserted target-level ``glue code.''

In this paper, we present a novel framework for the design and verification of
sound language interoperability that follows an interoperation-after-compilation
strategy. Language designers specify what data can be converted between types of
the two languages via a convertibility relation $\tau_A \conv \tau_B$
(``$\tau_A$ is convertible to $\tau_B$'') and specify target-level glue code
implementing the conversions. Then, by giving a semantic model of
source-language types as sets of target-language terms, they can establish not
only the meaning of the source types, but also \emph{soundness of conversions}:
i.e., whenever $\tau_A \conv \tau_B$, the corresponding pair of conversions
(glue code) convert target terms that behave like $\tau_A$ to target terms that
behave like $\tau_B$, and vice versa. With this, they can prove semantic type
soundness for the entire system. We illustrate our framework via a series of
case studies that demonstrate how our semantic interoperation-after-compilation approach
allows us both to account for complex differences in language semantics and make
efficiency trade-offs based on particularities of compilers or targets.

\end{abstract}

%% 2012 ACM Computing Classification System (CSS) concepts
%% Generate at 'http://dl.acm.org/ccs/ccs.cfm'.

%% End of generated code

%% Keywords
%% comma separated list
  %% \keywords are mandatory in final camera-ready submission

%% \maketitle
%% Note: \maketitle command must come after title commands, author
%% commands, abstract environment, Computing Classification System
%% environment and commands, and keywords command.
\maketitle

\section{Introduction}\label{sec:intro}
All practical language implementations come with some way of interoperating with
code written in a different language, usually via a foreign-function interface
(FFI). This enables development of software systems with components written in
different languages, whether to support legacy libraries or different
programming paradigms. For instance, you might have a system with a
high-performance data layer written in Rust interoperating with business logic
implemented in OCaml. Sometimes, this interoperability is realized by targeting
a common platform (e.g., Scala \cite{odersky2005scalable} and Clojure
\cite{hickey2020history} for the JVM, or SML \cite{benton2004adventures} and F\#
\cite{syme2006leveraging} for .NET). Other times, it is supported by libraries
that insert boilerplate or ``glue code'' to mediate between the two languages
(such as the binding generator SWIG \cite{beazley96}, C->Haskell
\cite{chakravarty99}, OCaml-ctypes \cite{yallop18}, NLFFI \cite{blume01}, Rust's
bindgen \cite{bindgen}, etc). While interoperability can be achieved in other
ways---via the network, inter-process communication, or dispatching between
interpreters and compiled code---we focus in this paper on the case when both
languages are compiled to a shared intermediate or target language.

In 2007, \citet{matthews07} observed that while there were
numerous FFIs that supported interoperation between languages, there had been no
effort to study the \emph{semantics} of interoperability. They proposed a simple and
elegant system for abstractly modeling interactions between languages $A$
and $B$ by embedding the existing operational syntax and semantics into a multi-language
$AB$ and adding boundaries to mediate between the two. Specifically, a boundary
$^{\tau_A}\!\mathcal{AB}^{\tau_B}(\cdot)$ allows a term $\tlang{e_B}$ of type $\tau_B$ to
be embedded in an $A$ context that expects a term of type $\tau_A$, 
and likewise for the boundary $^{\tau_B}\mathcal{BA}^{\tau_A}(\cdot)$.
Operationally, the term
$^{\tau_A}\!\mathcal{AB}^{\tau_B}(\tlang{e_B})$ evaluates $e_B$ using the $B$-language semantics to
$^{\tau_A}\!\mathcal{AB}^{\tau_B}(\tlang{v_B})$ and then a type-directed conversion
takes the value $\tlang{v_B}$ of type $\tau_B$ to an $A$-language term of type $\tau_A$.
There are often interesting design choices
in deciding what conversions are available for a type, if any at all.  
One can then prove that the entire multi-language type system is sound by proving type safety for
the multi-language, which includes the typing rules of both the embedded  
languages and the boundaries. This multi-language framework
has inspired a significant amount of work on 
interoperability: between simple and dependently typed languages~\cite{osera12},
between languages with unrestricted and substructural
types~\cite{tov10,scherer18}, between a high-level functional language and
assembly~\cite{patterson17}, and between source and target languages of
compilers~\cite{ahmed11,perconti14,new16}.  

Unfortunately, while Matthews-Findler-style boundaries give an elegant, abstract
model for interoperability, they do not reflect reality. Indeed, a decade and a
half later, there is little progress on assigning semantics to real
multi-language systems. In the actual implementations we study, the source languages are
compiled to components in a common target and glue code is inserted at the
boundaries between them to account for different data representations or calling
conventions. While one could try to approach this problem by defining
source-level boundaries, building a compiler for the multi-language, and then
showing that the entire system is realized correctly, there are serious
downsides to this approach. One is that if the two languages differ
significantly, the multi-language may be significantly more than just an
embedding of the evaluation rules of both languages (c.f. our last case study,
as an implicitly garbage-collected language interoperating with a manually
managed language may need to make the garbage collection explicit). And that
doesn't even consider the fact that in practice, we usually have \emph{existing}
compiler implementations for one or both languages and wish to add (or extend)
support for interoperability. Here, language designers' understanding of what
datatypes \emph{should} be convertible at the source level very much depends on
how the sources are compiled and how data is (or could be) represented in the
target, all information that is ignored by the multi-language approach.
Moreover, certain conversions, even if possible, might be undesirable because
the glue code needed to realize \emph{safe} interoperability imposes too much
runtime overhead.
% AA: changed above to address AHW comment.
% \ahw{
%   This statement is too strong because it leaves out one of the 
%   major design constraints that we've been consdiering: 
%   we want a model for interop that works with \emph{existing implementations}.
%   Without that constraint, I don't think this argument is so compelling.
%   If I have source types $\tau_A$ and $\tau_B$ that totally should be convertible, 
%   I can construct a brand new compiler (or modify an existing one) to enable that conversion.
%   The point we're trying to make is that if you take the compilers as fixed parameters,
%   \emph{then} they guide what conversions are safe.
% } 

In this paper, we present a framework for the \emph{design} and
\emph{verification} of sound language interoperability, where both activities
are connected to the actual implementation (of compilers and conversions). At
the source, we still use Matthews-Findler-style boundaries, as our approach differs
not in the source syntax but rather that instead of proving operational
properties of that source, we instead prove semantic type soundness by defining
a model of \emph{source} types as sets of (or relations on) \emph{target} terms.
That is, the interpretation of a source type is the set of target terms that
\emph{behave as that type}. Guiding the design of these type interpretations are
the compilers.
% AA: will check on what exactly qualifies as "realizability model".
This kind of model, often called a realizability model, is not a new idea
% Typically, semantic type soundness is proved using a model of source types as 
% sets of (untyped) source
% terms~\cite{ahmed04:phd,ahmed10:fpcc-jrnl,ahmed05:substruct,jung18:rustbelt}.  
% Nonetheless, modeling source types as sets of target terms is not a new idea 
--- for instance, \citet{benton07:ppdp} and \citet{benton09:tldi} used such
models to prove type soundness, but their work was limited to a single source
language. By interpreting the types of two source languages as sets of terms in
a common target, we enable rich reasoning about interoperability.
% , and gain powerful new reasoning ability. For example, we can ask if a
% particular type in one language is the same as a type in the other language:
% this is true if each type's set of target terms is identical. 
Using the model, we can then give meaning to
a boundary $^{\tau_B}\!\mathcal{BA}^{\tau_A}(\cdot)$: 
there is a bit of target code that, when given a target term that is in
the model of the type $\tau_A$, results in a target term in the model of type
$\tau_B$. 

A realizability model is valuable not only for proving soundness,
but for reasoning about the \emph{design} of interoperability.
For example, we can ask if a
particular type in one language is \emph{the same} as a type in the other language. 
This is true if the same set of target terms inhabits both types, and in this case
conversions between the types should do nothing.
% for understanding the challenges of interoperability as it allows us to cis
% approach brings to the fore target-level reasoning that is critical in the
% design of interoperability systems. 
More generally, opportunities for efficient conversions may only become apparent
upon looking at how source types and invariants are represented (or realized) in the target.
Since interoperability is a design challenge, with tradeoffs just like any other---performance 
high among them---working with the ability to understand all the
pieces is a tremendous advantage.

\paragraph{Contributions} To demonstrate the use and benefits of our framework,
we present three case studies that illustrate different kinds of
challenges for interoperability. In each case, we compile to an untyped target
language. 
\begin{enumerate}[wide]
  \item \textbf{Shared-Memory Interoperability (\S \ref{sec:ref}):} We consider how mutable references can be
    exchanged between two languages and what properties must hold of stored data
    for aliasing to be safe.  We show that to avoid copying mutable data --- 
    without having to wrap references in guards or
    chaperones~\cite{strickland12} ---
    convertible reference types must be inhabitated by the \emph{very same} set of target terms.
  \item \textbf{Affine \& Unrestricted (\S \ref{sec:affeff}):}
    We consider how $\langref$, a standard functional language with mutable references, can
    interact with $\langatlc$, an affine language.  We
    show that affine code can be safely embedded in unrestricted code and vice versa
    by using runtime checks (only where necessary) to ensure that affine resources are used at most
    once.
  \item \textbf{Memory Management \& Polymorphism (\S \ref{sec:polylin}):}
    We consider how $\langsysf$, whose references are garbage collected,
    can interact with $\langlt$~\cite{ahmed07:L3}, a language that uses linear
    capabilities to support safe strong updates to a manually managed mutable
    heap. We demonstrate not only when memory can be moved between languages, but also
    a type-level form of interoperability that allows generics to be used with
    $\langlt$ (which lacks type polymorphism) without violating any invariants
    of either language.
\end{enumerate}

For each case study, we devise a novel realizability model.  An interesting
aspect of these models is that, since the target languages are untyped,
statically enforced source invariants must be captured using either dynamic
enforcement in target code or via invariants in the model. This demonstrates 
that our approach is viable even when working with existing target languages
without rich static reasoning principles.

We chose these three case studies to exercise our framework both in terms of
type system invariants (substructural types, polymorphism) but also properly
handling mutable state (exchanging pointers and garbage collection). Interesting
challenges for the future could include differences of control-flow and concurrency.

% Although we primarily focus on these models in service of interoperability,
% many of them are insightful in their own right.
%
% Moreover, using the binary models we can prove not only soundness of
% conversions, but also that convertibility satisfies 
% a stronger property, namely \emph{boundary cancellation}~\cite{}: that converting
% a term to the other language and back again will either error or  
% result in a term equivalent to the original. 

\vspace{0.2cm}
\noindent \emph{Definitions and proofs elided from this paper are provided in
  our technical appendix \cite{patterson22:semint-tr}.}

%%% Local Variables:
%%% mode: latex
%%% TeX-master: "paper"
%%% End:

\section{The Framework}\label{sec:framework}

\newcommand{\langA}{\rhlang{A}}
\newcommand{\Alang}[1]{\rhlang{#1}}
\newcommand{\langB}{\rllang{B}}
\newcommand{\Blang}[1]{\rllang{#1}}
\newcommand{\langT}{\tlang{T}}

Before diving into the case studies that serve as evidence of its efficacy, we
first describe, in step-by-step fashion, the framework for proving type
soundness in the presence of interoperability that is the primary contribution
of this paper. The inputs to the framework are two source languages, language
$\langA$ and language $\langB$, a target language $\langT$, and compilers
$\Alang{e}^{+} = \tlang{e}$ and $\Blang{e}^{+} = \tlang{e}$. This section serves
both as a roadmap of what is to come and a reference to return to. The first
two steps (\S\ref{framework:step1} and \S\ref{framework:step2}) must be
performed by the \textbf{designer} of the interoperability system, whereas the
last three (\S\ref{framework:step3}, \S\ref{framework:step4}, and
\S\ref{framework:step5}) should be performed by the \textbf{verifier} of the
system. Note that, as with type soundness, partial verification is still
potentially useful, and so the first two steps should be seen as what needs to
be implemented, and the last three as what should be aspired to, if not formally carried out.

\subsection{Boundary syntax}\label{framework:step1}

To include code from another language, the designer requires some way of invoking such
code. While there are various ways of doing this in real toolchains, here she adopts a
general approach based on a notion of \emph{language
  boundaries}.

If a language $\langA$ is to include code from language $\langB$, the $\langA$
designer should add a boundary form $\Alang{\convert{\tau_{A}}{\Blang{e}}}$.
This allows a term $\Blang{e} : \Blang{\tau_{B}}$ to be used in an $\langA$
context at type $\Alang{\tau_{A}}$, for some $\Alang{\tau_{A}}$ and
$\Blang{\tau_{B}}$. This boundary strategy is very general: it allows both
inline code, a strategy adopted by many FFI libraries for C, but also the more
typical import/export style of linking. In that case, what is compiled would be
an open term with a $\langB$ binding
$\Blang{f} : \Blang{\tau \rightarrow \tau'}$ free. Then, the use of the imported
term would be $\Alang{\convert{\tau_{A} \rightarrow \tau_{A}'}{\Blang{f}}}$ for
appropriate types $\Alang{\tau_{A}}$ and $\Alang{\tau_{A}'}$.

Note that while in our examples, we equip both languages with boundaries, the framework
does not require this.

\subsection{Convertibility rules}\label{framework:step2}

To know whether a term $\Alang{\convert{\tau_{A}}{\Blang{e}}}$ is well-typed, the designer needs to know if
a $\langB$ term $\Blang{e} : \Blang{\tau_{B}}$ can be converted to an $\langA$ type
$\Alang{\tau_{A}}$. There is no way to know, a priori, what types can be converted, and
thus the framework requires that the designer specify this explicitly. In particular, she must provide
judgments of the form $\Alang{\tau_{A}} \conv \Blang{\tau_{B}}$ to indicate that these two types
are interconvertible, allowing for the possibility of dynamic conversion errors.
Since our notion of linking depends upon both language $\langA$ and $\langB$ being
compiled to a common target $\langT$, this conversion needs to be witnessed by $\langT$
code that performs the conversion. $C_{\Alang{\tau_{A}}\mapsto
  \Blang{\tau_{B}}}$ denotes the code that performs a target-level conversion
from $\Alang{\tau_{A}}$ to $\Blang{\tau_{B}}$. For example, if
$\Alang{bool} \conv \Blang{int}$, and the former compiles to the integers
$0$ and $1$, then the conversion $C_{\Alang{bool}\mapsto \Blang{int}}$ is
a no-op (since compiled booleans are already $\langT$ language integers), but
$C_{\Blang{int}\mapsto \Alang{bool}}$ must do something different. It
could raise a dynamic conversion error if given a $\langT$ $\tlang{int}$ other than $0$
or $1$, or it could collapse all other numbers into one of those, or something
else. The particular choice depends on the languages in question, and what the
designer of the interoperability system thinks makes sense: the framework only
requires that the decision made preserves type soundness.

\subsection{Realizability models for both languages}\label{framework:step3}

In order to prove type soundness, and in particular, account for the boundaries
and convertibility rules from \S\ref{framework:step1} and
\S\ref{framework:step2}, the verifier needs to build a logical relation for both languages.
This relation is atypical in two ways.
First, it is a \emph{realizability} model, which means that while it
is indexed by source types, it is inhabited by target terms. That is, the verifier must first
define an interpretation of values for each source type $\Alang{\tau}$, written
$\Vrel{\Alang{\tau}}{}$, as the set of $\langT$ language \emph{values}
$\tlang{v}$ that behave as $\Alang{\tau}$. That is, $\Vrel{\Alang{bool}}{}$ is
not the set of $\langA$ language booleans (i.e., $\Alang{true}$ and
$\Alang{false}$), but rather, the $\langT$ values that behave as $\langA$
booleans (likely, $0$ and $1$). In particular, the compiler from $\langA$ to
$\langT$ must send $\Alang{true}$ and $\Alang{false}$ into $\Vrel{\Alang{bool}}{}$, but the latter can
include more values. There is also an expression relation, written
$\Erel{\Alang{\tau_{A}}}{}$, that is the set of $\langT$ language terms that
evaluate to values in $\Vrel{\Alang{\tau_{A}}}{}$ (or diverge, or run to a
well-defined error). The second atypical, and novel, aspect is that the relation is indexed with
the types of \emph{both} of our source languages; in this example, $\langA$ and
$\langB$. Since they compile to the same target, this works: the
inhabitants of $\Vrel{\Alang{bool}}{}$ and $\Vrel{\Blang{int}}{}$ are both
$\langT$ values. By bringing the types of both languages into a common setting,
the verifier gains powerful reasoning principles; for example, we can ask if
$\Vrel{\Alang{bool}}{}$ = $\Vrel{\Blang{int}}{}$.

\subsection{Soundness of conversions}\label{framework:step4}

Using the realizability models defined in \S\ref{framework:step3}, the verifier can prove
that the convertibility rules defined in \S\ref{framework:step2} are sound. In
particular, if $\Alang{\tau_{A}} \conv \Blang{\tau_{B}}$, then she should show that the
conversions $C_{\Alang{\tau_{A}}\mapsto\Blang{\tau_{B}}}$ and $C_{\Blang{\tau_{B}}\mapsto\Alang{\tau_{A}}}$
actually translate expressions between the types correctly. This is done by
showing for any term $\tlang{e}$ in $\Erel{\Alang{\tau_{A}}}{}$, that 
$C_{\Alang{\tau_{A}}\mapsto\Blang{\tau_{B}}}(\tlang{e})$ is in $\Erel{\Blang{\tau_{B}}}{}$, and similarly for
$C_{\Blang{\tau_{B}}\mapsto\Alang{\tau_{A}}}$. Since the model defines type interpretations,
this ensures that the conversions do exactly what is expected.

\subsection{Soundness of entire languages}\label{framework:step5}

Proving the conversions sound (\S\ref{framework:step4}) is the central goal, of
course, but the verifier also needs to ensure that the model defined in
\S\ref{framework:step3} is actually faithful to the languages. She does this by
following the standard approach for proving semantic type soundness. First, for
each typing rule in both source languages, she proves that a corresponding lemma
holds in terms of the model. For example, for pairs she proves that if
$\tlang{e} \in \Erel{\Alang{\tau_{1} \times \tau_{2}}}{}$ then
$\Alang{fst}^{+}\tlang{~e} \in \Erel{\Alang{\tau_{1}}}{}$---note we write
$\Alang{fst}^{+}$, which is $\langT$ code (and could be an array projection, or some other $\langT$
operation), since what is in $\Erel{\Alang{\tau_{1}}}{}$ are $\langT$ terms.

%%% Local Variables:
%%% mode: latex
%%% TeX-master: "paper"
%%% End:

\section{Shared Memory}\label{sec:ref}

Aliased mutable data is challenging to deal with no matter the context, but
aliasing across languages is especially difficult because giving a pointer to a
foreign language can allow for \emph{unknown} data to be written to its address.
Specifically, if the pointer has a particular type in the host language,
then only certain data should be written to it, 
but the foreign language may not respect or even know about that restriction.
One existing approach to this problem is to create proxies, where data is guarded or
converted before being read or
written~\cite{dimoulas12,strickland12,mates19:refcc}. While sound,
this comes with significant runtime overhead.
Here, our framework allows a different approach.

\paragraph{Languages} In this case study, we explore this problem using two
simply-typed functional source languages with dynamically allocated mutable
references, $\langrhl$ and $\langrll$ (for ``higher-level'' and
``lower-level''). $\langrhl$ has boolean, sum, and product types, whereas
$\langrll$ has arrays ($\rllang{[e_1,\ldots,e_n]} : \rllang{[\tau]})$. Their
syntax is given in Fig.~\ref{fig:ref-syntax} and their static semantics ---
which are entirely standard --- are elided (see \cite{patterson22:semint-tr}).
These two languages are compiled (Fig.~\ref{fig:ref-compilers}---note that we write $e^{+}$ to indicate $e'$, where $e \rightsquigarrow e'$) into an untyped
stack-based language called $\langstack$ (inspired by \cite{Kleffner17}), whose
syntax and small-step operational semantics --- a relation on configurations
$\sconf{\heap}{\sstack}{\tlang{P}}$ comprised of a heap, stack, and program ---
are given in Fig.~\ref{fig:ref-target}; here we describe a few highlights.
First, we note that $\langstack$ values include not only numbers, thunks, and
locations, but arrays of values, a simplification we made for the sake of
presentation. Second, notice the interplay between $\tlang{thunk}$ and
$\tlang{lam}$: $\tlang{thunk}$s are suspended computations, whereas
$\tlang{lam}$ is an instruction (not a value) responsible solely for
substitution\footnote{À la Levy's Call-by-push-value \cite{levy01:phd}.}. We can see how these features are combined, or used separately, in
our compilers (Fig.~\ref{fig:ref-compilers}). Finally, note that for any
instruction where the precondition on the stack is not met, the configuration
steps to a program with $\tyfail$ (a dynamic type error), although we elide most of these reduction rules for
space.

\begin{figure}
  \begin{small}
    \[
      \begin{array}{llcl}
        \langrhl &

        \text{Type~} \rhlang{\tau} & ::= & \rhlang{unit} \vo \rhlang{bool} \vo \rhlang{\tau {+} \tau} \vo
        \rhlang{\tau {\times} \tau} \vo \rhlang{\tau {\rightarrow} \tau} \vo \rhlang{ref\,\tau}\\

        & \text{Expr.~} \rhlang{e} & ::= & \rhlang{()} \vo \rhlang{true} \vo \rhlang{false} \vo \rhlang{x} \vo
        \rhlang{inl~e} \vo \rhlang{inr~e}  \\
        & & & \hspace{-1cm}
        \vo
        \rhlang{(e,e)} \vo \rhlang{fst~e} \vo \rhlang{snd~e} \vo
        \rhlang{if~e~e~e} \vo \rhlang{\lambda x : \tau.e} \vo
        \rhlang{e~e}   \\
        & & & \hspace{-1cm}\vo \rhlang{match~e~x\{e\}~y\{e\}} \vo \rhlang{ref~e} \vo \rhlang{!e}
        \vo \rhlang{e := e} \vo \rhlang{\convert{\tau}{\rllang{e}}}\\

        \langrll&

        \text{Type~} \rllang{\tau} & ::= & \rllang{int} \vo \rllang{[\tau]}
        \vo \rllang{\tau \rightarrow \tau} \vo \rllang{ref~\tau}\\
        & \text{Expr.~} \rllang{e} & ::= & \rllang{n} \vo \rllang{x} \vo
        \rllang{[e,\ldots]} \vo \rllang{e[e]} \vo \rllang{\lambda x : \tau.e}
        \vo \rllang{e~e}  \\
        & & & \hspace{-1cm}  \vo \rllang{e + e} \vo \rllang{if0~e~e~e} \vo \rllang{ref~e} \vo \rllang{!e} \vo \rllang{e :=
          e} \vo
        \rllang{\convert{\tau}{\rhlang{e}}}
      \end{array}
    \]
    \caption{Syntax for $\langrhl$ and $\langrll$.}
    \label{fig:ref-syntax}
  \end{small}
\end{figure}

\begin{figure}
  \begin{small}
    \[
      \begin{array}{l@{\quad}c@{\quad}l}
        \text{Program}~ \tlang{P} &::=& \tlang{\cdot} \vo \tlang{i, P} \quad
        \text{Value}~ \tlang{v} ~~~::=~~~ \tlang{n} \vo \tlang{thunk~P} \vo
        \tlang{\loc} \vo \tlang{[v,\ldots]}\\
        \text{Instruction}~ \tlang{i} & ::= & \tlang{push~v} \vo \tlang{add} \vo
        \tlang{less?} \vo \tlang{if0~P~P} \vo \tlang{lam~x.P} \vo \tlang{call}
         \\
        & & \vo \tlang{idx} \vo \tlang{len} \vo \tlang{alloc} \vo \tlang{read} \vo
        \tlang{write} \vo \tlang{\fail{c}} \\
        \text{Error Code}~ \tlang{c} &::=& \tlang{\tyerrcode} \vo \tlang{\idxerrcode} \vo \tlang{\converrcode}\\
        \text{Heap}~\heap &::=& \{\ell \!:\! \tlang{v}, \ldots\} \qquad
        \text{Stack}~\sstack ~~~::=~~~ \tlang{v}, \ldots, \tlang{v} \vo \tlang{Fail~c} \\
      \end{array}
    \]
    \[\arraycolsep=0.25pt
      \begin{array}[t]{lcll}
        \sconf{\heap}{\sstack}{\tlang{push~v, P}} 
        &\rightarrow&
        \sconf{\heap}{\sstack, \tlang{v}}{\tlang{P}} 
        & (\sstack\neq \tlang{Fail~c})\ \\

        \sconf{\heap}{\sstack, \tlang{n', n}}{\tlang{add, P}} 
        &\rightarrow&
        \sconf{\heap}{\sstack, \tlang{(n + n')}}{\tlang{P}}
        & \\

        \sconf{\heap}{\sstack, \tlang{n', n}}{\tlang{less?, P}} 
        &\rightarrow&
        \sconf{\heap}{\sstack, \tlang{b}}{\tlang{P}}
        & \hspace{-0.75cm}(\tlang{b}\!=\!\tlang{0}\hspace{2.15pt}\text{if}\hspace{2.15pt}\tlang{n}\!<\!\tlang{n'}\hspace{2.15pt}\text{else}\hspace{2.15pt}\tlang{1}) \\

        \sconf{\heap}{\sstack, \tlang{n}}{\tlang{if0~P_1~P_2, P}} 
        &\rightarrow&
        \sconf{\heap}{\sstack}{\tlang{P_i , P}}
        & \hspace{-0.75cm}(\tlang{i}\!=\!\tlang{1}\hspace{2.15pt}\text{if}\hspace{2.15pt}\tlang{n}\!=\!\tlang{0}\hspace{2.15pt}\text{else}\hspace{2.15pt}\tlang{2}) \\

        \sconf{\heap}{\sstack}{\tlang{if0~P_1~P_2, P}}
        &\rightarrow&
        \sconf{\heap}{\sstack}{\tlang{\tyfail}}
        & (\sstack \neq \sstack', \tlang{n}) \\

        \sconf{\heap}{\sstack, \tlang{v}}{\tlang{lam~x.P_1, P_2}} 
        &\rightarrow& 
        \sconf{\heap}{\sstack}{\tlang{[x \!\mapsto\! v]P_1, P_2}}
        & \\
        \sconf{\heap}{\sstack, \tlang{thunk~P_1}}{\tlang{call, P_2}} 
        &\rightarrow& 
        \sconf{\heap}{\sstack}{\tlang{P_1, P_2}}
        & \\

        \sconf{\heap}{\sstack, \tlang{[v_0, \ldots, v_{n'}], n}}{\tlang{idx, P}} 
        &\rightarrow& 
        \sconf{\heap}{\sstack, \tlang{v_{n}}}{\tlang{P}}
        & (\tlang{n}\!\in\![\tlang{0}, \tlang{n'}]) \\

        \sconf{\heap}{\sstack, \tlang{[v_0, \ldots, v_{n'}], n}}{\tlang{idx, P}} 
        &\rightarrow& 
        \sconf{\heap}{\sstack}{\tlang{\idxfail}}
        & (\tlang{n}\!\notin\![\tlang{0}, \tlang{n'}]) \\

        \sconf{\heap}{\sstack, \tlang{[v_0, \ldots, v_n]}}{\tlang{len, P}} 
        &\rightarrow& 
        \sconf{\heap}{\sstack, \tlang{(n + 1)}}{\tlang{P}}
        & \\

        \sconf{\heap}{\sstack, \tlang{v}}{\tlang{alloc, P}} 
        &\rightarrow& 
        \sconf{\heap \!\uplus\! \{\tlang{\ell \!:\! v}\}}{\sstack, \tlang{\ell}}{\tlang{P}}
        & \\

        \sconf{\heap \!\uplus\! \{\tlang{\ell \!:\! v}\}}{\sstack, \tlang{\ell}}{\tlang{read, P}} 
        &\rightarrow& 
        \sconf{\heap \!\uplus\! \{\tlang{\ell \!:\! v}\}}{\sstack, \tlang{v}}{\tlang{P}}
        & \\

        \sconf{\heap \!\uplus\! \{\tlang{\ell \!:\! \_}\}}{\sstack, \tlang{\ell, v}}{\tlang{write, P}} 
        &\rightarrow& 
        \sconf{\heap \!\uplus\! \{\tlang{\ell \!:\! v}\}}{\sstack}{\tlang{P}}
        & \\

        \sconf{\heap}{\sstack}{\tlang{\fail{c}, P}}
        &\rightarrow&
        \sconf{\heap}{\tlang{Fail~c}}{\cdot}
        & \\
      \end{array}
    \]
    \caption{Syntax and selected operational semantics for $\langstack$ (most
      $\tyfail$ cases elided).}
    \label{fig:ref-target}
  \end{small}
\end{figure}

\begin{figure}
  \begin{small}
    \[
      \begin{array}{c}
        \tlang{SWAP \triangleq lam~x.(lam~y.push~x,push~y)}\\
        \tlang{DROP \triangleq lam~x.()} \qquad
        \tlang{DUP \triangleq lam~x.(push~x,push~x)}
      \end{array}
    \]
    \[
      \begin{array}[t]{lcl}
    \rhlang{()} \rightsquigarrow \tlang{push~0} & \mid &
    \rhlang{x} \rightsquigarrow \tlang{push~x}\\
    \rhlang{true} \mid \rhlang{false} & \rightsquigarrow & \tlang{push~\langle 0 \mid 1\rangle}\\
    \rhlang{inl~e} \mid \rhlang{inr~e} & \rightsquigarrow & \tlang{\rhlang{e^\tlang{+}},lam~x.(push~[\langle 0 \mid 1\rangle,x])}\\
    \rhlang{if~e~e_1~e_2} & \rightsquigarrow &
    \tlang{\rhlang{e^\tlang{+}},if0~\rhlang{e_1^\tlang{+}}~\rhlang{e_2^\tlang{+}}}\\
    \rhlang{match~e} & \rightsquigarrow & \tlang{\rhlang{e^\tlang{+}},DUP,push~1,idx,SWAP,push~0,} \\
    \enspace \rhlang{~x\{e_1\}~y\{e_2\}} & & 
    \enspace \tlang{idx,if0~(lam~x.\rhlang{e_1^\tlang{+}})~(lam~y.\rhlang{e_2^\tlang{+}})}\\

    \rhlang{(e_1,e_2)} & \rightsquigarrow &
    \tlang{\rhlang{e_1^\tlang{+}},\rhlang{e_2^\tlang{+}},lam~x_2,x_1.(push~[x_1,x_2])}\\
    \rhlang{fst~e} \mid \rhlang{snd~e} & \rightsquigarrow & \tlang{\rhlang{e^\tlang{+}},push~\langle 0 \mid 1\rangle,idx}\\
    % \rhlang{snd~e}  & \rightsquigarrow & \tlang{\rhlang{e^\tlang{+}},push~1,idx}\\

    % \rhlang{\lambda x : \tau.e} & \rightsquigarrow & \tlang{push~(thunk~lam~x.\rhlang{e^\tlang{+}})}\\
    \rhlang{e_1~e_2} & \rightsquigarrow & \tlang{\rhlang{e_1^\tlang{+}},\rhlang{e_2^\tlang{+}},SWAP,call}\\
    \rhlang{ref~e} & \rightsquigarrow & \tlang{\rhlang{e^\tlang{+}},alloc}\\
    % \rhlang{!e} & \rightsquigarrow & \tlang{\rhlang{e^\tlang{+}},read}\\
    \rhlang{e_1 := e_2} & \rightsquigarrow & \tlang{\rhlang{e_1^\tlang{+}},\rhlang{e_2^\tlang{+}},write,push~0}\\
    \rhlang{\convert{\tau}{\rllang{e}}} & \rightsquigarrow &
    \rllang{e^\tlang{+}},C_{\rllang{\tau}{\mapsto}\rhlang{\tau}}\\
    \rllang{n} \rightsquigarrow \tlang{push~n} & \mid &
    \rllang{e_1+e_2} \rightsquigarrow \tlang{\rllang{e_1^\tlang{+}},\rllang{e_2^\tlang{+}},SWAP,add}\\
    % \rllang{x} & \rightsquigarrow & \tlang{push~x}\\
    % \rllang{if0~e~e_1~e_2} & \rightsquigarrow &
    % \tlang{\rllang{e^\tlang{+}},if0~\rllang{e_1^\tlang{+}}~\rllang{e_2^\tlang{+}}}\\
    \rllang{[e_1,\ldots,e_n]} & \rightsquigarrow &
    \tlang{\rllang{e_1^\tlang{+}},\ldots,\rllang{e_n^\tlang{+}},lam~x_n,\ldots,x_1.} \\
    & & \enspace \tlang{(push~[x_1,\ldots, x_n])}\\
    \rllang{e_1[e_2]} & \rightsquigarrow & \tlang{\rllang{e_1^\tlang{+}},\rllang{e_2^\tlang{+}},idx}\\
    \rllang{\lambda x : \tau.e} & \rightsquigarrow & \tlang{push~(thunk~lam~x.\rllang{e^\tlang{+}})}\\
    % \rllang{e_1~e_2} & \rightsquigarrow &
    % \tlang{\rllang{e_1^\tlang{+}},\rllang{e_2^\tlang{+}},SWAP,call}\\
    % \rllang{ref~e} & \rightsquigarrow & \tlang{\rllang{e^\tlang{+}},alloc}\\
    \rllang{!e} & \rightsquigarrow & \tlang{\rllang{e^\tlang{+}},read}\\
    % \rllang{e_1 := e_2} & \rightsquigarrow & \tlang{\rllang{e_1^\tlang{+}},\rllang{e_2^\tlang{+}},write,push~0}\\
    \rllang{\convert{\tau}{\rhlang{e}}} & \rightsquigarrow &
    \rhlang{e^\tlang{+}},C_{\rhlang{\tau}{\mapsto}\rllang{\tau}}
  \end{array}
    \]
    \caption{Selections from compilers for $\langrhl$ and $\langrll$.}
    \label{fig:ref-compilers}
  \end{small}
\end{figure}

\paragraph{Convertibility} In our source languages, we may syntactically embed a
term from one language into the other using the boundary forms 
$\rhlang{\convert{\tau_A}{\rllang{e}}}$ and
$\rllang{\convert{\tau_B}{\rhlang{e}}}$. The typing rules for boundary terms
require that the boundary types be convertible, 
written $\rhlang{\tau_A} \conv \rllang{\tau_B}$. Those typing rules are:

\begin{small}
\begin{mathpar}
  \inferrule{\rllang{\Gamma};\rhlcolor{\Gamma} \vdash \rhlang{e} : \rhlang{\tau_A} \\
    \rhlang{\tau_A} \conv \rllang{\tau_B}}{\rhlcolor{\Gamma};\rllang{\Gamma} \vdash \rllang{\convert{\tau_B}{\rhlang{e}}} : \rllang{\tau_B}}

  \inferrule{\rhlcolor{\Gamma};\rllang{\Gamma} \vdash \rllang{e} : \rllang{\tau_B} \\
     \rhlang{\tau_A} \conv \rllang{\tau_B}}{\rllang{\Gamma};\rhlcolor{\Gamma} \vdash
    \rhlang{\convert{\tau_A}{\rllang{e}}} : \rhlang{\tau_A}}
\end{mathpar}
\end{small}

Note that the convertibility judgment is a declarative, extensible judgment that
describes closed types in one language that are interconvertible with closed
types in the other, allowing for the possibility of well-defined runtime errors.
By separating this judgment from the rest of the type system, the language
designer can allow additional conversions to be added later, whether by
implementers or even end-users. The second thing to note is that this
presentation allows for open terms to be converted, so we must maintain a type
environment for both languages during typechecking (both $\rhlcolor{\Gamma}$ and
$\rllang{\Gamma}$), as we have to carry information from the site of
binding---possibly through conversion boundaries---to the site of variable use.
A simpler system, which we have explored, would only allow closed terms to be
converted. In that case, the typing rules still use the
$\rhlang{\tau_A} \conv \rllang{\tau_B}$ judgment but do not thread foreign
environments (using only $\rhlcolor{\Gamma}$ for $\langrhl$ and only
$\rllang{\Gamma}$ for $\langrll$).

We present, in Fig.~\ref{fig:ref-conv}, some of the convertibility rules we
have defined for this case study (we elide $\rhlang{\tau_1 \times \tau_2} \conv \rllang{[\tau]}$), which come with target-language instruction
sequences that perform the conversions, written $C_{\rhlang{\tau_A}{\mapsto}\rllang{\tau_B}}$ (some are no-ops). An
instruction sequence $C_{\rhlang{\tau_A}{\mapsto}\rllang{\tau_B}}$, while ordinary
target code, when appended to a program in the model at type $\rhlang{\tau_A}$,
should result in a program in the model at type $\rllang{\tau_B}$. An
implementer can write these conversions based on understanding
of the sets of target terms that inhabit each source type, before defining a
proper semantic model (or possibly, without defining one, if formal soundness is
not required). They would do this based on inspection of the compiler and the
target. 

From Fig.~\ref{fig:ref-compilers}, we see that $\rhlang{bool}$ and
$\rllang{int}$ both compile to target integers, and importantly, that
$\rhlang{if}$ compiles to $\tlang{if0}$, which 
means the compiler interprets $\rhlang{false}$ as any non-zero
integer. Hence, conversions between $\rhlang{bool}$ and $\rllang{int}$ are identities.

For sums, we use the tags $\tlang{0}$ and $\tlang{1}$, and as
for $\rhlang{if}$, we use $\tlang{if0}$ to branch in the compilation of
$\rhlang{match}$. Therefore, we can choose if the $\rhlang{inl}$ and
$\rhlang{inr}$ tags should be represented by $0$ and $1$, or by $0$ and any other
integer $\tlang{n}$. Given that tags could be added later, we
choose the former, thus converting a sum to an array of integers is
mostly a matter of converting the payload. In the other direction, we have to
handle the case that the array is too short, and error.

The final case, between $\rhlang{ref~bool}$ and
$\rllang{ref~int}$, is the reason for this case study.
Intuitively, if you exchange pointers, any value of the new type can now be
written at that address, and thus must have been compatible with the old type
(as aliases could still exist). Thus, we require that $\rhlang{bool}$ and
$\rllang{int}$ are somehow ``identical'' in the target, so conversions are unnecessary.

\begin{figure}
\begin{small}
\begin{mathpar}
  \inferrule{ }{
    \tlang{C_{\rhlang{bool} \mapsto \rllang{int}}},
    \tlang{C_{\rllang{int} \mapsto \rhlang{bool}}}
    :
    \rhlang{bool}
    \conv
    \rllang{int}
  }
  \and
  \inferrule{ }{
    \tlang{C_{\rhlang{ref~bool} \mapsto \rllang{ref~int}}},
    \tlang{C_{\rllang{ref~int} \mapsto \rhlang{ref~bool}}}
    :
    \rhlang{ref~bool}
    \conv
    \rllang{ref~int}
  }
  \and
  \inferrule{
    \tlang{C_{\rhlang{\tau_1} \mapsto \rllang{int}}},
    \tlang{C_{\rllang{int} \mapsto \rhlang{\tau_1}}}
    :
    \rhlang{\tau_1}
    \conv
    \rllang{int}
    \\
    \tlang{C_{\rhlang{\tau_2} \mapsto \rllang{int}}},
    \tlang{C_{\rllang{int} \mapsto \rhlang{\tau_2}}}
    :
    \rhlang{\tau_2}
    \conv
    \rllang{int}
  }{
    \tlang{C_{\rhlang{\tau_1 + \tau_2} \mapsto \rllang{[int]}}},
    \tlang{C_{\rllang{[int]} \mapsto \rhlang{\tau_1 + \tau_2}}}
    :
    \rhlang{\tau_1 + \tau_2}
    \conv
    \rllang{[int]}
  }
\end{mathpar}
\[
  \tlang{C_{\rhlang{bool} \mapsto \rllang{int}}} \triangleq
  \tlang{C_{\rllang{int} \mapsto \rhlang{bool}}} \triangleq
  \tlang{C_{\rhlang{ref~bool} \mapsto \rllang{ref~int}}} \triangleq
  \tlang{C_{\substack{\rllang{ref} \rllang{int}} \mapsto
      \substack{\rhlang{ref}\\\rhlang{bool}}}} \triangleq \cdot
\]
\[
    \begin{array}[t]{l}
      \tlang{C_{\rhlang{\tau_1 + \tau_2} \mapsto \rllang{[int]}}} \triangleq
         \tlang{DUP,~push~1,~idx,~SWAP,}
        \tlang{push~0,~idx,~DUP,} \\
        \hspace{2cm} \tlang{if0~(SWAP,~C_{\rhlang{\tau_1} \mapsto \rllang{int}})}\\
        \hspace{2.4cm}\tlang{(SWAP,~C_{\rhlang{\tau_2} \mapsto \rllang{int}}),}
        \tlang{lam~x_v.lam~x_t.push~[x_t, x_v]}\\

      \tlang{C_{\rllang{[int]} \mapsto \rhlang{\tau_1 + \tau_2}}} \triangleq\\
      \quad \tlang{DUP,~len,~push~2,~SWAP,}
        \tlang{less?,~if0~\convfail,} \\
        \quad \tlang{DUP,push~1,~idx,~SWAP,}
        \tlang{push~0,idx,~DUP,} \\
        \quad \tlang{if0~(SWAP,~C_{\rllang{int} \mapsto \rhlang{\tau_1}})}
        \tlang{\big(DUP,push~{-1},~add,}\\
        \quad\enspace \tlang{if0~(SWAP,~C_{\rllang{int} \mapsto \rhlang{\tau_2}})}
        \tlang{(\convfail)\big),} \tlang{lam~x_v.lam~x_t.push~[x_t, x_v]}
    \end{array}
  \]
\end{small}
\caption{Conversions for $\langrhl$ and $\langrll$.}
\label{fig:ref-conv}
\end{figure}

\paragraph{Semantic Model} Declaring that a type $\rhlang{bool}$ is
``identical'' to $\rllang{int}$ or that $\rhlang{\tau}$ is convertible
to $\rllang{\tau}$ and providing the conversion code is not sufficient for
soundness. In order to show that these conversions are sound, and indeed to
understand which conversions are even possible, we define a model for source
types that is inhabited by target terms. Since both languages compile to the
same target, the range of their relations will be the same (i.e., composed of
terms and values from $\langstack$), and thus we will be able to easily and
directly compare the inhabitants of two types, one from each language.

Our model, which aside from the use of $\langstack$ is a standard step-indexed unary logical relation for a
language with mutable state (essentially following~\citet{ahmed04:phd}), is
presented with some parts elided in Fig.~\ref{fig:ref-lr} (see \cite{patterson22:semint-tr}).

We give value interpretations for each source type $\tau$, written
$\Vrel{\tau}{}$ as sets of target \emph{values} $\tlang{v}$ paired with
\emph{worlds} $\W$ that inhabit that type. A world $\W$ is comprised of a step
index $k$ and a \emph{heap typing} $\heapty$, which maps locations to type
interpretations in $Typ$. As is standard, $Typ$ is the set of valid type
interpretations, which must be closed under world extension. A future world
$\W'$ extends $W$, written $W'\sqsupseteq W$, if $W'$ has a potentially lower
step budget $j \leq \W.k$ and all locations in $W.\heapty$ still have the same
types (to approximation $j$). 

Intuitively, $(\W,\tlang{v}) \in \Vrel{\tau}{}$ says
that the target value $\tlang{v}$ belongs to (or behaves like a value of) type
$\tau$ in world $W$. For example, $\Vrel{\rhlang{unit}}{}$ is inhabited by
$\tlang{0}$ in any world. A more interesting case is $\Vrel{\rhlang{bool}}{}$,
which is the set of all target integers, not just $\tlang{0}$ and $\tlang{1}$,
though we could choose to define our model that way (provided we compiled
$\rhlang{bool}$s to $\tlang{0}$ or $\tlang{1}$). An array
$\Vrel{\rllang{[\tau]}}{}$ is inhabited by an array of target values
$\tlang{v_i}$ in world $\W$ if each $\tlang{v_i}$ is in $\Vrel{\rllang{\tau}}{}$
with $\W$.

Functions follow the standard pattern for logical relations, appropriately
adjusted for our stack-based target language:
$\Vrel{\rhlang{\tau_1\rightarrow\tau_2}}{}$ is inhabited by values
$\tlang{thunk~lam~x.P}$ in world $\W$ if, for any future world $\W'$ and argument
$\tlang{v}$ in $\Vrel{\rhlang{\tau_1}}{}$ at that world, the result of substituting the
argument into the body ($\tlang{[x{\mapsto}v]P}$) is in the expression relation
at the result type $\Erel{\rhlang{\tau_2}}{}$. Reference types
$\Vrel{\rhlang{ref~\tau}}{}$ are inhabited by a location $\loc$ in world $\W$ if
the current world's heap typing $\W.\heapty$ maps $\loc$ to the value relation
$\Vrel{\rhlang{\tau}}{}$ approximated to the step index in the world
$\W.k$. (The $j$-approximation of a type, written $\floor{j}{\Vrel{\tau}{}}$, restricts
$\Vrel{\tau}{}$ to inhabitants with worlds in $World_{j}$.)

Our expression relation $\Erel{\tau}{}$ defines when a program $\tlang{P}$
in world $\W$ behaves as a computation of type $\tau$. It says that for any heap
$\heap$ that satisfies the current world $\W$, written $H : \W$, and any non-$\tlang{Fail}$ stack
$\sstack$, if the machine $\sconf{\heap}{\sstack}{\tlang{P}}$ terminates in $j$
steps (where $j$ is less than our step budget $\W.k$), then either it ran to a non-type
error or there exists some value 
$\tlang{v}$ and some future world $\W'$ such that the resulting stack $\sstack'$
is the original stack with $\tlang{v}$ on top, the resulting heap $\heap'$
satisfies the future world $\W'$ and 
$\W'$ and $\tlang{v}$ are in $\Vrel{\tau}{}$.

% Our worlds, as is standard, are made up of step indexes $k$ and heap typings
% $\heapty$, where the latter maps heap locations $\loc$ to sets in $Typ$. The
% latter is made up of sets of pairs of worlds and values that are closed under world extension.

At the bottom of Fig.~\ref{fig:ref-lr}, we show a syntactic shorthand,
$\sem{\rhlcolor{\Gamma}; \rllang{\Gamma}  \vdash \rllang{e} : \rllang{\tau}}$,
for showing that well-typed source programs, when compiled and closed off with
well-typed substitutions $\stack$ that map variables to target values, are in the
expression relation. Note $\Envrel{\Gamma}{}$ contains
closing substitutions $\stack$ in world $W$ that assign every $x : \tau \in
\Gamma$ to a $\tlang{v}$ such that $(W,\tlang{v}) \in \Vrel{\tau}{}$.

With our logical relation in hand, we can now state formal properties about our
convertibility judgments.

\begin{lemma}[Convertibility Soundness]\label{lemma:ref-conv-soundness}~

  If $\rhlang{\tau} \conv \rllang{\tau}$, then
  $\forall (\W,P) \!\in\! \Erel{\rhlang{\tau}}{}.
  (\W,(P,C_{\rhlang{\tau}{\mapsto}\rllang{\tau}})) \!\in\! \Erel{\rllang{\tau}}{}
  $ $\land~ \forall (\W,P) \!\in\! \Erel{\rllang{\tau}}{}.
  (\W,(P,C_{\rllang{\tau}{\mapsto}\rhlang{\tau}})) \!\in\! \Erel{\rhlang{\tau}}{}$.
\end{lemma}

\begin{proof}
  We sketch the $\rhlang{ref~bool}\conv\rllang{ref~int}$ case; (rest elided, see \cite{patterson22:semint-tr}). For
  $\rhlang{ref~bool}\conv\rllang{ref~int}$, what we need to show is that given
  any expression in $\Erel{\rhlang{ref~bool}}{}$, if we apply the conversion
  (which does nothing), the result will be in $\Erel{\rllang{ref~int}}{}$. That
  requires $\Vrel{\rhlang{ref~bool}}{} =
  \Vrel{\rllang{ref~int}}{}$.

  The value relation at a reference type says that if you look up the location
  $\loc$ in the heap typing of the world ($\W.\heapty$), you will get the value
  interpretation of the type. That means a $\rhlang{ref~bool}$ must be a
  location $\tlang{\ell}$ that, in the model, points to the value interpretation
  of $\rhlang{bool}$ (i.e., $\Vrel{\rhlang{bool}}{}$). In our model, this must
  be true for all future worlds, which makes sense for ML-style references. Thus, for this proof to
  go through, $\Vrel{\rhlang{bool}}{}$ must be the same as
  $\Vrel{\rllang{int}}{}$, which it is.
\end{proof}

Once we have proved Lemma~\ref{lemma:ref-conv-soundness}, we can prove semantic
type soundness in the standard two-step way for our entire system. First, for
each source typing rule, we define a compatibility lemma that is a semantic
analog to that rule. For example, the compatibility lemma for the conversion
typing rule, shown here, requires the proof of
Lemma~\ref{lemma:ref-conv-soundness} to go through:

\[
  \sem{
    \rhlcolor{\Gamma}; \rllang{\Gamma}  \vdash
    \rllang{e} : \rllang{\tau}
  } \land
  \rhlang{\tau} \conv \rllang{\tau}
  \implies
  \sem{
    \rllang{\Gamma}; \rhlcolor{\Gamma} \vdash
    \rhlang{\rhlang{\convert{\tau}{\rllang{e}}}} : \rhlang{\tau}
  }
\]

Once we have all compatibility lemmas we can prove the following theorems as a consequence:

\begin{theorem}[Fundamental Property]\label{theorem:ref-fundamental}~

  If $~\rhlcolor{\Gamma}; \rllang{\Gamma}  \vdash
  \rllang{e} : \rllang{\tau}$ then $\sem{\rhlcolor{\Gamma}; \rllang{\Gamma}  \vdash
    \rllang{e} : \rllang{\tau}}$ and   if $~\rllang{\Gamma}; \rhlcolor{\Gamma}  \vdash
  \rhlang{e} : \rhlang{\tau}$ then $\sem{\rllang{\Gamma}; \rhlcolor{\Gamma}  \vdash
    \rhlang{e} : \rhlang{\tau}}$.

\end{theorem}

\begin{theorem}[Type Safety for $\langrll$]\label{theorem:ref-safety-ll}{~}
  If $\rhlcolor{\cdot}; \rllang{\cdot}  \vdash
  \rllang{e} : \rllang{\tau}$ then for any $\heap : \W$,
  if $\sconf{\heap}{\cdot}{\rllang{e}^+}\steps{*} \sconf{\heap'}{\sstack'}{\tlang{P'}}$,
  then either $\sconf{\heap'}{\sstack'}{\tlang{P'}}\step \sconf{\heap''}{\sstack''}{\tlang{P''}}$, or $\tlang{P'}=\cdot$ and either $\sstack'=\tlang{Fail~c}$ for some $\tlang{c}\in\inlokerr$ or $\sstack'=\tlang{v}$.
\end{theorem}

\begin{theorem}[Type Safety for $\langrhl$]\label{theorem:ref-safety-hl}{~}
  If $\rllang{\cdot}; \rhlcolor{\cdot}  \vdash
  \rhlang{e} : \rhlang{\tau}$ then for any $\heap : \W$,
  if $\sconf{\heap}{\cdot}{\rhlang{e}^+}\steps{*} \sconf{\heap'}{\sstack'}{\tlang{P'}}$,
  then either $\sconf{\heap'}{\sstack'}{\tlang{P'}}\step \sconf{\heap''}{\sstack''}{\tlang{P''}}$, or $\tlang{P'}=\cdot$ and either $\sstack'=\tlang{Fail~c}$ for some $\tlang{c}\in\inlokerr$ or $\sstack'=\tlang{v}$.
\end{theorem}

\begin{figure}
  \begin{small}
    \[
      \begin{array}{l}
      AtomVal_n = \{ (\W,\tlang{v}) \vo \W \in World_{n} \} \\[0.5em]
      \mathit{World}_n = \{(k,\heapty) \vo k < n \land \heapty \subset
      \mathit{HeapTy}_k\} \\[0.5em]
      \mathit{HeapTy}_n = \{\loc \mapsto \mathit{Typ}_{n}, \ldots \}\\[0.5em]
      Typ_n = \{ R \in 2^{AtomVal_n} \vo \forall (\W,\tlang{v}) \in
      R.~ \\ \qquad \qquad\forall\W'.~\W
      \worldext \W' \implies (\W',\tlang{v}) \in R  \}   \\[0.5em]
      \end{array}
    \]
    \[\arraycolsep=1.5pt
      \begin{array}[t]{rcl}
        \Vrel{\rhlang{bool}}{} &=& \{(\W,\tlang{n}) \} \qquad \Vrel{\rhlang{unit}}{} ~~=~~ \{(\W,\tlang{0}) \}\\
        \Vrel{\rhlang{\tau_1 + \tau_2}}{} &=& \{(\W, \tlang{[0,v]})
        \vo (\W,\tlang{v}) \in \Vrel{\rhlang{\tau_1}}{} \} \\
        & & \enspace \cup~
        \{(\W,\tlang{[1,v]} )
        \vo  (\W,\tlang{v})
        \in \Vrel{\rhlang{\tau_2}}{} \}\\
        \Vrel{\rhlang{\tau_1 \rightarrow \tau_2}}{} &=&
        \{(\W,\tlang{thunk~lam~x.P}) \vo \\
        & & \enspace \forall \tlang{v},\W' \sqsupset \W.~
        (\W',\tlang{v}) \in \Vrel{\rhlang{\tau_1}}{} \\ 
        & & \quad \implies (\W', \tlang{[x{\mapsto}v]P}) \in \Erel{\rhlang{\tau_2}}{}
        \}\\
        \Vrel{\rhlang{ref~\tau}}{} &=& \{(\W,\tlang{\loc})
        \vo
        \W.\heapty(\loc) = \floor{\W.k}{\Vrel{\rlang{\tau}}{}} \}\\
        \Vrel{\rllang{int}}{} &=& \{(\W,\tlang{n})\}\\
        \Vrel{\rllang{[\tau]}}{} &=& \{(\W, \tlang{[v_1,\ldots,v_n]}) \vo (\W,
        \tlang{v_i}) \in \Vrel{\rllang{\tau}}{} \}\\
        \Vrel{\rllang{\tau_1 \rightarrow \tau_2}}{} &=&
        \{(\W,\tlang{thunk~lam~x.P}) \vo \\
        & & \enspace \forall \tlang{v},\W' \sqsupset \W.~
        (\W',\tlang{v}) \in \Vrel{\rllang{\tau_1}}{} \\
        & & \quad \implies (\W', \tlang{[x{\mapsto}v]P}) \in \Erel{\rllang{\tau_2}}{}
        \}\\
        \Vrel{\rllang{ref~\tau}}{} &=&  \{(\W,\tlang{\loc})
        \vo
        \W.\heapty(\loc) = \floor{\W.k}{\Vrel{\rllang{\tau}}{}} \}\\
      \end{array}
    \]
    \[
      \begin{array}{l}
        \Erel{\tau}{} = \{(\W,P) \vo
        \forall \heap {:} \W,S \neq \tlang{Fail~\_},\heap',S',j < \W.k.~
         \\
         \qquad
         \sconf{\heap}{\sstack}{P} \steps{j} \sconf{\heap'}{\sstack'}{\cdot}\implies
        \sstack' = \tlang{Fail~c} \land \tlang{c} \in \tlang{\inlokerr}\\
        \qquad\lor~
        \exists \tlang{v},\W' \sqsupseteq W.
        ~\big(
        \sstack' = \sstack, \tlang{v} \land
        \heap' : \W'
        \land  (\W', \tlang{v}) \in \Vrel{\tau}{})\big)\}\\[0.5em]
        \sem{\rllang{\Gamma};\rhlcolor{\Gamma} \vdash \rhlang{e} : \rhlang{\tau}} \equiv
        \forall \W\,\stack_{\rllang{\Gamma}}\, \stack_{\rhlcolor{\Gamma}}\,.
        (\W,\stack_{\rllang{\Gamma}}) \in
        \Envrel{\rllang{\Gamma}}{} \land (\W,\stack_{\rhlcolor{\Gamma}}) \in
        \Envrel{\rhlcolor{\Gamma}}{}\\
        \hspace{2cm} \implies (\W,\close(\stack_{\rllang{\Gamma}},\close(\stack_{\rhlcolor{\Gamma}},\rhlang{e}^+))) \in
        \Erel{\rhlang\tau}{}\\[0.5em]
        \sem{\rhlcolor{\Gamma};\rllang{\Gamma} \vdash \rllang{e} : \rllang{\tau}} \equiv
        \forall \W\,\stack_{\rhlcolor{\Gamma}}\, \stack_{\rllang{\Gamma}}\,.
        (\W,\stack_{\rhlcolor{\Gamma}}) \in
        \Envrel{\rhlcolor{\Gamma}}{} \land (\W,\stack_{\rllang{\Gamma}}) \in
        \Envrel{\rllang{\Gamma}}{}\\
        \hspace{2cm} \implies (\W,\close(\stack_{\rhlcolor{\Gamma}},\close(\stack_{\rllang{\Gamma}},\rllang{e}^+))) \in
        \Erel{\rllang\tau}{}
  \end{array}
    \]
    \caption{Logical relation for $\langrhl$ and $\langrll$.}
    \label{fig:ref-lr}
  \end{small}
\end{figure}

\paragraph{Discussion}
In addition to directly passing across pointers, there are two alternative
conversion strategies, both of which our framework would accommodate. First, we
could create a new location and copy and convert the data. This would allow the more flexible
convertibility which does not require references to ``identical'' types, but
would not allow aliasing, which may be desirable. Second, we could
convert $\rhlang{(unit\rightarrow \tau) \times (\tau \rightarrow unit)}$
and $\rllang{(unit\rightarrow \tau) \times (\tau \rightarrow unit)}$ instead $\rhlang{ref~\tau}$ and $\rllang{ref~\tau}$ (assuming
we had pairs)---i.e., read/write proxies to the reference (similar to that used
in \cite{dimoulas12}). This allows aliasing,
i.e., both languages reading / writing to the same location, and is
sound for arbitrary convertibility relations, but comes at a runtime
cost at each read / write.

The choice to use the encoding described in this case study, or either of these
options, is not exclusive---we could provide different options for
different types in the same system, depending on the performance characteristics we need.

%%% Local Variables:
%%% mode: latex
%%% TeX-master: "paper"
%%% End:

% \input{case-affdyn.tex}
\section{Affine \& Unrestricted}\label{sec:affeff}

In our second case study, we consider an affine language, $\langatlc$,
interacting with an unrestricted one, $\langref$. We enforce $\langatlc$'s
at-most-once variable use dynamically in the
target using the well-known technique described, e.g., in \cite{tov10}, where
affine resources are protected behind thunks with stateful flags that raise
runtime errors the second time the thunk is forced. However, an interesting and
challenging aspect of our case study is that we only want to use dynamic
enforcement when we  lack static assurance that an affine variable will be use
at most once. 
% This allows more flexible sound mixing than would be possible if the types were linear.

\paragraph{Languages} We present the syntax of $\langatlc$, $\langref$, and our
untyped Scheme-like functional target $\langtarget$ in
Fig.~\ref{fig:affeff-syntax} and selected static semantics in
Fig.~\ref{fig:affeff-statics} (see supplementary material \cite{patterson22:semint-tr}). Our target $\langtarget$ is untyped, with functions, pattern
matching, mutable references, and a standard operational semantics defined via
steps $\tconf{\heap}{\tlang{e}}\steps{}\tconf{\heap'}{\tlang{e'}}$ over heap and
expression pairs. As in the previous case study, we will support 
open terms across language boundaries, and thus need to carry environments for
both languages throughout our typing judgments.

\begin{figure}[t!]
  \begin{small}
    \[
      \arraycolsep=1pt
      \begin{array}{lcl}
        \langatlc \\
        \text{Type~}\alang{\tau} & ::= & \alang{unit \vo bool \vo int \vo \tau
          {\dynlolli} \tau \vo \tau {\statlolli} \tau \vo ~!\tau}  \vo
        \alang{\tau \& \tau \vo \tau {\otimes} \tau} \\
        \text{Expr.~}\alang{e} & ::= & \alang{() \vo true \vo false \vo n \vo
          x \vo a_\modal \vo \lambda a_{\modal} : \tau . e}\\ & & \vo \alang{e
          ~ e}  \vo \alang{\lconvert}\rlang{e}\alang{\rconvert_{\tau}}
        \alang{\vo ~!v} \alang{ \vo let~ !x = e ~in~ e' \vo \langle e, e'
          \rangle}
        \\ & & \alang{\vo e.1 \vo e.2 \vo (e,e) \vo let~ (a_\stat,a_\stat') = e ~in~ e'} \\
        \text{Value~}\alang{v} & ::= & \alang{() \vo \lambda a_\modal : \tau . e \vo ~!v \vo \langle e, e'
          \rangle \vo (v,v')}\\
        \text{Mode~}\alang{\modal} & ::= & \alang{\dyn} \vo \alang{\stat} \\
        \langref\\
        \text{Type~} \rlang{\tau} & ::= & \rlang{unit} \vo  \rlang{int} \vo \rlang{\tau {\times} \tau}
        \vo \rlang{\tau {+} \tau} \vo \rlang{\tau
          {\rightarrow} \tau} \vo \rlang{\forall\alpha.\tau \vo \alpha}
        \vo \rlang{ref\,\tau}\\

        \text{Expr.~} \rlang{e} & ::= & \rlang{()} \vo \rlang{n} \vo \rlang{x} \vo \rlang{(e,e)} \vo
        \rlang{fst~e} \vo \rlang{snd~e}  \vo
        \rlang{inl~e} \vo \rlang{inr~e} \\
        & & \vo
        \rlang{match~e~x\{e\}~y\{e\}}  \vo \rlang{\lambda x : \tau.e} \vo \rlang{e~e} \vo
        \rcolor{\Lambda\alpha.\mathtt{e}} \vo \rlang{e{[}\tau{]}} \\ & & \vo
        \rlang{ref~e} \vo \rlang{!e} \vo \rlang{e := e} \vo
        \rlang{\convert{\tau}{\alang{e}}}\\
        \langtarget \\
        \text{Expr}~ \tlang{e} & ::= & \tlang{()} \vo \tlang{n}
        \vo \glang{\ell} \vo
        \tlang{x} \vo \tlang{(e,e)} \vo
        \tlang{fst~e} \vo \tlang{snd~e}  \vo
        \tlang{inl~e} \vo \tlang{inr~e} \\ & &
         \vo
        \tlang{if~e~\{e\}~\{e\}} \vo      \tlang{match~e~x\{e\}~y\{e\}}
        \vo \tlang{let~x=e~in~e} \\
        & & \vo
        \tlang{\lambda x \{e\}} \vo \tlang{e~e}  \vo
        \tlang{ref~e}  \vo \tlang{!e} \vo
        \tlang{e := e} \vo \tlang{fail~c} \\
        \text{Values} ~ \tlang{v} & ::= & \tlang{()} \vo \tlang{n} \vo
        \glang{\ell} \vo \tlang{(v,v)} \vo
        \tlang{\lambda x.e} \\
        \text{Err}~\tlang{c} & ::= & \tyerrcode \vo \converrcode \\
      \end{array}
    \]

\vspace{-0.25cm}
    \caption{Syntax for $\langref$, $\langatlc$, and $\langtarget$.}
    \label{fig:affeff-syntax}
  \end{small}
\end{figure}

% A non-standard aspect of $\langatlc$, the reader may have noticed are the two
% kinds of affine function types $\alang{\dynlolli}$ and $\alang{\statlolli}$.
To avoid unnecessary dynamic enforcement, we have two
kinds of affine function types in $\langatlc$: $\alang{\dynlolli}$ and
$\alang{\statlolli}$.\footnote{In our supplementary materials \cite{patterson22:semint-tr}, we also present a complete case study with a
simpler variant of $\langatlc$, which does not distinguish
$\alang{\dynlolli}/\alang{\statlolli}$ and thus does dynamic enforcement even on
affine variables that have no interaction with unrestricted code.}
We introduce a distinction between $\langatlc$ functions (and thus
bindings) that may be passed across the boundary (our ``dynamic'' affine arrows
$\alang{\dynlolli}$, written with a hollow circle and bind dynamic
affine variables $\alang{a_{\dyn}}$), and ones that will only ever be used
within $\langatlc$ (our ``static'' affine arrows $\alang{\statlolli}$,
 written with a solid circle and bind static affine variables $\alang{a_{\stat}}$).

We can see in Fig.~\ref{fig:affeff-statics} how $\langatlc$'s affine-variable
environment $\alang{\Affenv}$ is maintained: variables are introduced by lambda
and tensor-destructuring let, and environments are split across subterms, but all
bindings are not required to be used, as we can see, in the variable rule. (In
the full rules in supplementary material \cite{patterson22:semint-tr}, a similar pattern shows up for base
types). Since affine resources can exist within unrestricted $\langref$ terms,
our affine environments $\alang{\Affenv}$ need to be split, even in $\langref$
typing rules.

\begin{figure}
  \begin{small}
    \begin{mathpar}
      \inferrule{\alang{a_\modal} : \alang{\tau} \in \alang{\Affenv} }{\rcolor{\Delta};\rcolor{\Gamma};\alang{\Gamma};\alang{\Affenv} \vdash \alang{a_\modal} :
        \alang{\tau}}
      \and
\inferrule{\rcolor{\Delta};\rcolor{\Gamma};\alang{\Gamma};\alang{\Affenv[a_\dyn := \tau_1]} \vdash \alang{e} :
  \alang{\tau_2} \\ \text{no}_{\acolor\stat}(\alang{\Affenv})}{\rcolor{\Delta};\rcolor{\Gamma};\alang{\Gamma};\alang{\Affenv}\ \vdash \alang{\lambda a_\dyn : \tau_1. e} :
  \alang{\tau_1 \dynlolli \tau_2}}
\and
\inferrule{\rcolor{\Delta};\rcolor{\Gamma};\alang{\Gamma};\alang{\Affenv[a_\stat := \tau_1]} \vdash \alang{e} :
  \alang{\tau_2}}{\rcolor{\Delta};\rcolor{\Gamma};\alang{\Gamma};\alang{\Affenv}\ \vdash \alang{\lambda a_\stat : \tau_1. e} :
  \alang{\tau_1 \statlolli \tau_2}}
\and
\inferrule{\alang{\Affenv} = \alang{\Affenv\alang{_1} \uplus
    \alang{\Affenv}\alang{_2}} \\
  \rcolor{\Delta};\rcolor{\Gamma};\alang{\Gamma};\alang{\Affenv}\alang{_1} \vdash
  \alang{e_1} : \alang{\tau_1 \modalolli
    \tau_2} \\ \rcolor{\Delta};\rcolor{\Gamma};\alang{\Gamma};\alang{\Affenv}\alang{_2} \vdash \alang{e_2} :
  \alang{\tau_1}}{\rcolor{\Delta};\rcolor{\Gamma};\alang{\Gamma};\alang{\Affenv} \vdash \alang{e_1 ~ e_2}  : \alang{\tau_2}}
\and
\inferrule{\alang{\Affenv} = \alang{\Affenv_1} \uplus \alang{\Affenv_2}
  \\ \rcolor{\Delta};\rcolor{\Gamma};\alang{\Gamma};\alang{\Affenv}\alang{_1} \vdash \alang{e} : \alang{\tau_1 \otimes \tau_2} \\
  \rcolor{\Delta};\rcolor{\Gamma};\alang{\Gamma};\alang{\Affenv_2[a_\stat := \tau_1, a_\stat' := \tau_1]}
  \vdash \alang{e'} : \alang{\tau'}}{\rcolor{\Delta};\rcolor{\Gamma};\alang{\Gamma};\alang{\Affenv} \vdash
  \alang{let~ (a_\stat,a_\stat') = e ~in~ e'} : \alang{\tau'}}
\and
\inferrule{\alang{\Affenv}=\alang{\Affenv_e}\uplus\alang{\Affenv'} \\ \text{no}_{\acolor\stat}(\alang{\Affenv_e}) \\ \rcolor{\Delta};\rcolor{\Gamma};\alang{\Gamma};\alang{\Affenv_e} \vdash \alang{e} : \alang{\tau} \\ \_ : \alang{\tau} \conv
  \rlang{\tau}}{\alang{\Gamma};\alang{\Affenv};\rcolor{\Delta};\rcolor{\Gamma} \vdash \rlang{\convert{\tau}{\alang{e}}} :
  \rlang{\tau}}
\and
\inferrule{\alang{\Affenv} = \alang{\Affenv\alang{_1} \uplus
    \alang{\Affenv}\alang{_2}} \\
  \alang{\Gamma};\alang{\Affenv}\alang{_1};\rcolor{\Delta};\rcolor{\Gamma} \vdash
  \rlang{e_1} : \rlang{\tau_1 \rightarrow
    \tau_2} \\ \alang{\Gamma};\alang{\Affenv}\alang{_2};\rcolor{\Delta};\rcolor{\Gamma} \vdash \rlang{e_2} :
  \rlang{\tau_1}}{\rcolor{\Delta};\rcolor{\Gamma};\alang{\Gamma};\alang{\Affenv} \vdash \rlang{e_1 ~ e_2}  : \rlang{\tau_2}}
\end{mathpar}
    \caption{Selected statics for $\langatlc$ and $\langref$.}
    \label{fig:affeff-statics}
  \end{small}
\end{figure}

Note that we do not allow a dynamic function $\alang{\lambda a_\dyn : \_. e}$ to
close over static resources, as it may be duplicated if passed to $\langref$,
and thus the static resources would be unprotected.
However, we do allow a dynamic function 
to accept a static closure as argument.
This is safe because the dynamic guards will ensure that the static closure is called
at most once.
Once called, any static resources in its body will be used safely
because the static closure typechecked.

We present selections of our compilers in Fig.~\ref{fig:affeff-compiler} that
highlight the interesting cases: how we compile variables, binders, and
application. In the application cases, we can see that static variables do not
introduce the overhead that dynamic variables have (see the $\thnk$
macro at the top of the figure that errors on second invocation).

\begin{figure}
  \begin{small}
\[
  \thnk(\tlang{e}) \triangleq \tlang{let~r_{fr} = ref~1~in~}\lambda \_.\{
  \tlang{if~!r_{fr}~\{\convfail\}~\{r_{fr}:=0;e\}\}}
\]
\[
  \rlang{()} \rightsquigarrow \tlang{()}\quad
  \rlang{n} \rightsquigarrow  \tlang{n}\quad
  \rlang{\lambda x : \tau.e} \rightsquigarrow \tlang{\lambda x.\{\rlang{e}^+\}}\quad
    \alang{true}/\alang{false} \rightsquigarrow \tlang{0}/\tlang{1}
\]
\[
  \alang{a_\dyn}  \rightsquigarrow  \tlang{a_\dyn~()} \quad     \alang{a_\stat} \rightsquigarrow  \tlang{a_\stat} \quad \alang{\lambda a_{\dyn/\stat} : \tau . e} \rightsquigarrow \tlang{\lambda a_{\dyn/\stat}. \{\alang{e^\tlang{+}}\}}
\]
\[
  \begin{array}[t]{lcl}
    \alang{(e_1 : \tau_1 \dynlolli \tau_2) ~ e_2}  & \rightsquigarrow &
                                                                        \tlang{\alang{e_1^\tlang{+}}~(let~x = \alang{e_2^\tlang{+}}~in~\thnk(x))}\\
   \alang{(e_1 : \tau_1 \statlolli \tau_2) ~ e_2}  & \rightsquigarrow & \tlang{\alang{e_1}^\tlang{+}~\alang{e_2}^\tlang{+}}\\

    \alang{let~ (a_\stat,a_\stat') = e_1 ~in~ e_2} & \rightsquigarrow & \tlang{let~x_{\text{fresh}}
                                                         =
                                                         \alang{e_{1}}^\tlang{+},~} \\ & &
                                                         \tlang{a_\stat = fst~x_{\text{fresh},}} \\ & &
                                                         \tlang{a_\stat' =
                                                                                                        snd~x_{\text{fresh}}~in~\alang{e
                                                                                                        _{2}}^\tlang{+}}\\
  \end{array}
\]
    \caption{Selected cases for $\langref$ and $\langatlc$ compilers.}
    \label{fig:affeff-compiler}
  \end{small}
\end{figure}

\paragraph{Convertibility}

We define convertibility relations and conversions for $\langatlc$ and
$\langref$, highlighting selections in Fig.~\ref{fig:affeff-conv} (see
supplementary material for elided $\alang{unit} \conv \rlang{unit}$ and
$\alang{\tau_1 \otimes \tau_2} \conv \rlang{\tau_1 \times \tau_2}$). We focus on
the conversion between $\rlang{\rightarrow}$ and $\alang{\dynlolli}$ (note, of
course, that it is impossible to safely convert $\alang{\statlolli}$ to
$\langref$). Our compiler is designed to support affine code being
mixed directly with unrestricted code. 
Intuitively, an affine function should be able to behave as an unrestricted one,
but the other direction is harder to accomplish, and higher-order functions mean
both must be addressed at once. In order to account for this, we convert
$\alang{\tau_1 \multimap \tau_2}$ not to $\rlang{\tau_1 \rightarrow \tau_2}$, but rather to
$\rlang{(unit \rightarrow \tau_1) \rightarrow \tau_2}$. That is, to a $\langref$
function that expects its argument to be a thunk containing a $\rlang{\tau_1}$
rather than a $\rlang{\tau_1}$ directly. Provided that the thunk fails if
invoked more than once, we can ensure, dynamically, that a $\langref$ function
with that type behaves as an $\langatlc$ function of a related type. These
invariants are ensured by appropriate wrapping and use of the compiler macro
$\thnk(\cdot)$ (see top of Fig.~\ref{fig:affeff-compiler}).

\begin{figure}
  \begin{small}
    \begin{mathpar}
      \inferrule{
      }{C_{\rlang{int}{\mapsto}\alang{bool}},C_{\alang{bool}{\mapsto}\rlang{int}}
        : \rlang{int} \conv \alang{bool}}

      \inferrule{        C_{\alang{\tau_1}{\mapsto}\rlang{\tau_1}},C_{\rlang{\tau_1}{\mapsto}\alang{\tau_1}}
        :\alang{\tau_1} \conv \rlang{\tau_1} \\ C_{\rlang{\tau_2}{\mapsto}\alang{\tau_2}},C_{\alang{\tau_2}{\mapsto}\rlang{\tau_2}}
        :\alang{\tau_2} \conv \rlang{\tau_2}}{
        C_{\alang{\_}{\mapsto}\rlang{\_}},C_{\rlang{\_}{\mapsto}\alang{\_}}
        :\alang{\tau_1 \multimap \tau_2} \conv \rlang{
          (unit \rightarrow \tau_1) \rightarrow \tau_2}}
    \end{mathpar}

    \[
     C_{\alang{bool}{\mapsto}\rlang{int}}(\tlang{e}) \triangleq \tlang{e}
    \qquad
    C_{\rlang{int}{\mapsto}\alang{bool}}(\tlang{e}) \triangleq \tlang{if~e~0~1}
\]
\[\arraycolsep=2pt
  \begin{array}{l}
     \tlang{C}_{\alang{\tau_1 \multimap \tau_2}{\mapsto}\rlang{(unit \rightarrow \tau_1) \rightarrow
    \tau_2}}\tlang{(e)} \triangleq \\
    \quad \tlang{let~x=e~in~\lambda
       x_{thnk}.let~x_{conv} = \tlang{C}_{\rlang{\tau_1}{\mapsto}\alang{\tau_1}}(x_{thnk}~())~in}
     \\ \hspace{2.7cm} \tlang{let~x_{acc}=\thnk(x_{conv})
      ~in~\tlang{C}_{\alang{\tau_2}{\mapsto}\rlang{\tau_2}}(x~x_{acc})}\\

    \tlang{C}_{\rlang{(unit \rightarrow \tau_1) \rightarrow \tau_2}{\mapsto}\alang{\tau_1 \multimap
          \tau_2}}\tlang{(e)}
                                                     \triangleq \tlang{let~x=e~in~}\\ \quad \tlang{
          \lambda x_{thnk}. let~x_{acc}=\thnk(C_{\alang{\tau_1}\mapsto\rlang{\tau_1}}(x_{thnk}~()))~in~C_{\rlang{\tau_2}{\mapsto}\alang{\tau_2}}(\tlang{x~x_{acc}})}\\
    \end{array}
  \]
    \caption{Selected convertibility rules for $\langref$ and $\langatlc$.}
    \label{fig:affeff-conv}
  \end{small}
\end{figure}

\paragraph{Semantic Model} The most interesting part of this case study is the
logical relation because we must build a model that allows us to show that the
dynamic and static affine bindings within $\langatlc$ are used at most once. For
a dynamic binding, this is tracked in target code by the dynamic reference flag
created by the macro $\tlang{thunk}$. For a static binding, we use a similar
strategy of tracking use via a flag, but rather than a target-level dynamic runtime flag, we create a \emph{phantom} flag
that exists only within our model.
% challenging to deal with since they are a resource that must be kept
% isolated from parts of the program that would be allowed to use them more than
% once. To reason about this in our model, 
Specifically, we define an augmented target operational semantics that exists
solely for the model, and any program that runs without getting stuck under the
augmented semantics has a trivial erasure to a program that runs under the
standard semantics. This means we are using the model to identify a subset of
target programs (the erasures of well-behaved augmented programs) that do not violate source type constraints (i.e., do not use static
variables more than once), even if there is nothing in the target programs that actually
witnesses those constraints (i.e., dynamic checks or static types).

We build the model as follows. First, we extend our machine configurations to keep track
of \emph{phantom flags} $\flag$ --- i.e., in addition to a heap $\heap$ and term
$\tlang{e}$, we have a \emph{phantom flag set} $\fstore$. Second, the augmented
semantics uses one additional term, $\prtct$, which consumes one of the
aforementioned phantom flags when it reduces:

\begin{small}
\[
  \begin{array}{c}
    \text{Expressions}~ \tlang{e} ::= \ldots \tlang{\prtct(e,\flag)}\\[0.5em]
    \phconf{\fstore\uplus\{\flag\}}{\heap}{\tlang{\prtct(e,\flag)}} \phstep \phconf{\fstore}{\heap}{\tlang{e}}
\end{array}
\]
\end{small}

\noindent And finally, we modify the two rules that introduce bindings such that whenever
a binding in the syntactic category $\stat$ is introduced, we create a new 
phantom flag (where ``$\flag$ fresh'' means $\flag$ is disjoint from all flags
generated in this execution):

\begin{small}
\begin{mathpar}
  \begin{array}{l}
  \inferrule{\flag\text{ fresh}}
  {\phconf{\fstore}{\heap}{\tlang{let~a_\stat = v~in~e}} {\phstep} \phconf{\fstore\uplus\{\flag\}}{\heap}{\tlang{[a_\stat{\mapsto} \prtct(v, \flag)]e}}}\\
  \inferrule{\flag\text{ fresh}}
  {\phconf{\fstore}{\heap}{\tlang{\lambda a_\stat. e ~ v}} {\phstep}
    \phconf{\fstore\uplus\{\flag\}}{\heap}{\tlang{[a_\stat{\mapsto} \prtct(v,
        \flag)]e}}}
\end{array}
\end{mathpar}
\end{small}

\noindent Note that we write $\phstep$ for a step in this augmented semantics, to
distinguish it from the true operational step $\step$. 
While phantom flags in the augmented operational semantics play a similar role 
in protecting static affine resources as dynamic reference flags in the dynamic case,
the critical difference is that in the augmented semantics, a
$\tlang{\prtct(\cdot)}$ed resource for which there is no phantom flag will get
stuck, and thus be excluded from the logical relation by construction. This is very different from
the dynamic case, where we want --- and, in fact, need --- to include terms that can 
fail in order to mix $\langref$ and $\langatlc$ without imposing an
affine type system on $\langref$ itself.  What this means for the model is that
dynamic reference flags are a \emph{shared resource} that can be accessed from many
parts of the program and therefore tracked in the world, while phantom flags are
an \emph{unique resource} which our type system ensures is owned/used by at most 
one part of the program, which is what allows us to prove that the augmented
semantics will not get stuck.

While the full definitions are in our supplementary materials \cite{patterson22:semint-tr}, we give a high-level
description of our expression and value relations, shown in
Fig.~\ref{fig:affeff-lr}, noting that the high-level structure is similar to the
first case study.

% With those preliminary definitions in mind, we can define the expression
% relation, $\Erel{\tau}\rho$.

Our expression relation, $\Erel{\tau}\rho$, is made up of tuples of worlds
$\W$ and phantom flag stores / term pairs $(\fstore_i,e_i)$, where each
flag store represents the phantom variables owned by the expression. Our worlds
$\W$ keep the step index, a standard heap typing $\heapty$ (see
\S\ref{sec:ref}), but also an affine flag store $\astore$, which maps
dynamic flags $\loc$ to \emph{either} a marker that indicates a
dynamic affine variable has been used ($0$, written $\aused$), \emph{or} the
phantom flags $\fstore$ that it closes over if it has not been used (a set that can be empty, of course). These dynamic
flags $\loc$ are a subset of the heap, disjoint from $\heapty$
(which tracks the rest of the heap, i.e., all the normal/non-dynamic-flag
references). The expression relation
then says that, given a heap that satisfies the world and arbitrary ``rest'' of
phantom flag store $\fstore_{r}$ (disjoint from that closed over by the world
and the owned portion), the term $\tlang{e}$ will either: (i) run longer than
the step index accounts for, (ii) $\tlang{\convfail}$ (error while converting a
value), or (iii) terminate at some value $\tlang{e'}$, where the flag store
$\fstore$ has been modified to $\fstore_{f}\uplus\fstore_{g}$, the heap has
changed to $\heap'$, and the new world $\W'$ is an extension of $\W$. World
extension ($\worldext_{\fstore_{r}}$)
is defined over worlds that do not contain phantom flags from $\fstore_{r}$, since phantom flags are a local resource and the world contains what is global.
It allows the step index to
decrease, the heap typing to gain (but not overwrite or remove) entries,
and the affine store to mark (but not unmark) dynamic bindings as $\aused$.

At that future world, we know
that the resulting value, along with their $\fstore_{f}$, will be
in the value relation $\Vrel{\tau}\rho$. The phantom flag store $\fstore_{g}$
is ``garbage'' that is no longer needed, and the ``rest'' is unchanged. Note
that, while running, some phantom flags may have moved into the new
world but the new world cannot have absorbed what was in the ``rest''.

Our value relation cases are now mostly standard, so we will focus only on the
interesting ones: $\alang{\dynlolli}$ and
$\alang{\statlolli}$. $\Vrel{\alang{\tau_1\dynlolli \tau_2}}\cdot$ is defined to 
take an arbitrary argument from $\Vrel{\alang{\tau_1}}\cdot$, which may own static
phantom flags in $\fstore$, and add a new location
$\loc$ that will be used in the thunk that prevents multiple uses, but
also store the phantom flags in the affine store. The idea is that a function
$\alang{\lambda a_\dyn:\_.e}$ can be applied to an expression that closes over
static phantom flags, like $\alang{let~(b_\stat,c_\stat) = (1,2)~in~\lambda
  a_\stat.b_\stat}$---the latter will have phantom flags for both
$\alang{b_\stat}$ and $\alang{c_\stat}$. The body is then run with the
argument substituted with a guarded expression. Now, consider what happens when
the variable is used: the $\grd(\cdot)$ wrapper will update the location to
$\aused$, which means that in the world, the phantom flags that were put at that
location are no longer there --- i.e., they are no longer returned by $\rf(\W')$,
which returns all phantom flags closed over by dynamic flags.  That means, for
the reduction to be well-formed, the phantom flags have to move somewhere
else---either back to being owned by the term (in $\fstore_{f}$) or in the
discarded ``garbage'' $\fstore_{g}$.  Once the phantom flag set has been moved
back out of the world, the flags can again be used by $\prtct(\cdot)$
expressions.
%, as the phantom flag sets captured in the world are inaccessible. 

The static function, $\Vrel{\alang{\tau_1 \statlolli \tau_2}}{\cdot}$, has a
similar flavor, but it may itself own static phantom flags. That means that the
phantom flag set for the argument must be disjoint, and when we run the
body, we combine the set along with a fresh phantom flags
$\flag$ for the argument, which are then put inside the
$\prtct(\cdot)$ expressions.

\begin{figure}
  \begin{small}
\[
  \begin{array}{l@{~~}c@{~~}l}
    \multicolumn{3}{l}{\grd(\tlang{e}, \ell) \triangleq \lambda \_.\{
  \tlang{if~!\ell~\{\convfail\}~\{\ell:=\aused;e\}\}}}\\[0.5em]
 % \Vrel{\rlang{\mathtt{unit}}}{\rho} &=& \{(\W, (\emptyset, \tlang{()}), (\emptyset, \tlang{()})) \}\\
  % \Vrel{\rlang{\mathtt{int}}}{\rho} &=& \{(\W, (\emptyset, \tlang{n}), (\emptyset, \tlang{n})) \vo \tlang{n} \in \mathbb{Z} \}\\

  % \Vrel{\rlang{\tau_1 \times \tau_2}}{\rho} &=& \{(\W,
  % (\emptyset, (\tlang{v_{1a}},\tlang{v_{2a}})), (\emptyset, (\tlang{v_{1b}},\tlang{v_{2b}}))) \\
  % & & \quad \vo
  % (\W, (\emptyset, \tlang{v_{1a}}), (\emptyset, \tlang{v_{1b}})) \in
  % \Vrel{\rlang{\tau_1}}{\rho} \\ & & \qquad \land
  % (\W, (\emptyset, \tlang{v_{2a}}), (\emptyset, \tlang{v_{2b}})) \in \Vrel{\rlang{\tau_2}}{\rho}
  % \}\\

  % \Vrel{\rlang{\tau_1 + \tau_2}}{\rho} &=& \{(\W, (\emptyset, \tlang{inl~v_1}), (\emptyset, \tlang{inl~v_2}))\\
  % & & \quad \vo (\W, (\emptyset, \tlang{v_1}), (\emptyset, \tlang{v_2})) \in \Vrel{\rlang{\tau_1}}{\rho} \} \\
  % & & \hspace{-1.5cm} \cup~
  % \{(\W, (\emptyset, \tlang{inr~v_1}), (\emptyset, \tlang{inr~v_2}))
  % \vo  (\W, (\emptyset, \tlang{v_1}), (\emptyset, \tlang{v_2}))
  % \in \Vrel{\rlang{\tau_2}}{\rho} \}\\

  \Vrel{\rlang{\tau_1 \rightarrow \tau_2}}{\rho} &=&
    \{(\W, (\emptyset, \tlang{\lambda x.\{e\}})) \vo
    \forall \tlang{v}~ \W'.\\ & & \W \worldextstrict_{\emptyset} \W'~
    \land~ (\W', (\emptyset, \tlang{v})) \in
    \Vrel{\rlang{\tau_1}}{\rho} \\
    & & \implies (\W', (\emptyset, [x {\mapsto} \tlang{v}]\tlang{e})) \in \Erel{\rlang{\tau_2}}{\rho}
    \}\\

  % \Vrel{\rlang{\mathtt{ref}~\tau}}{\rho} &=& \{(\W, (\emptyset, \tlang{\loc_1}), (\emptyset, \tlang{\loc_2}))
  % \vo
  % \W.\heapty(\loc_1,\loc_2) = \floor{\W.k}{\Vrel{\rlang{\tau}}{\rho}} \}\\

  % \Vrel{\rlang{\forall\alpha.\tau}}{\rho} &=&
  %   \{(\W, (\emptyset, \tlang{\lambda.e_1}), (\emptyset, \tlang{\lambda.e_2})) \vo
  %   \forall R \in \UnrTyp,~\W'. \\ & & \W \worldextstrict_{\emptyset, \emptyset} \W'
  %   {\implies} (\W', (\emptyset, \tlang{e_1}), (\emptyset, \tlang{e_2})) \in \Erel{\rlang{\tau}}{\rho[\alpha \mapsto R]}\}\\

  % \Vrel{\rlang{\alpha}}{\rho} &=& \rho(\alpha)\\

  % \Vrel{\alang{unit}}{\cdot} &=& \{(\W, (\emptyset, \tlang{()}), (\emptyset,
  % \tlang{()}))\}\\
  % \Vrel{\alang{bool}}{\rho} &=& \{(\W,(\emptyset, 0),(\emptyset,0))\} \\ & &
  % \quad \cup~
  % \{(\W,(\emptyset, \tlang{n_1}),(\emptyset,\tlang{n_2}))
  % \vo n_1 \ne 0 \land n_2 \ne 0  \}\\
  % \Vrel{\alang{int}}{\cdot} &=& \{(\W, (\emptyset, \tlang{n}), (\emptyset, \tlang{n})) \vo  \tlang{n} \in \mathbb{Z}\}\\

  \Vrel{\alang{\tau_1 \dynlolli \tau_2}}{\cdot} &=&
  \{(\W, (\emptyset, \tlang{\lambda~x \{e\}})) \vo
 \forall \fstore~\tlang{v}~ \W'. \\ & & \quad \W \worldextstrict_{\emptyset} \W' \land
  (\W', (\fstore, \tlang{v})) \in \Vrel{\alang{\tau_1}}{\cdot} \\ & &
  \implies ((\W'.k, \W'.\heapty, \W'.\astore \uplus \ell \mapsto
  \fstore),\\ & & \hspace{0.85cm} (\emptyset, [x
  {\mapsto}\grd(\tlang{v},\ell)]\tlang{e})) \in \Erel{\alang{\tau_2}}{\cdot}
    \}\\

  \Vrel{\alang{\tau_1 \statlolli \tau_2}}{\cdot} &=&
    \{(\W, (\fstore, \tlang{\lambda~a_\stat. \{e\}})) \vo
    \\ & &  \forall \fstore'~\flag_1~\tlang{v}~ \W'. \W \worldextstrict_{\fstore} \W' \land~(\W', (\fstore', \tlang{v}))
    \in \Vrel{\alang{\tau_1}}{\cdot} \\
    & & \quad \land \fstore\cap \fstore'=\emptyset \land~\flag\notin \fstore\uplus \fstore'\uplus \rf(\W')\\ & &
    \implies (\W', (\fstore\uplus \fstore'\uplus \{\flag\}, [\tlang{a_\stat} {\mapsto} \tlang{\prtct(v,\flag)}]\tlang{e})) \\ & & \hspace{1.4cm}\in \Erel{\alang{\tau_2}}{\cdot}
    \}\\

  % \Vrel{\alang{!\tau}}{\cdot} &=& \{(\W, (\emptyset, \tlang{v_1}), (\emptyset, \tlang{v_2})) \vo (\W, (\emptyset, \tlang{v_1}), (\emptyset, \tlang{v_2})) \in \Vrel{\alang{\tau}}{\cdot}\}\\

  % \Vrel{\alang{\tau_1 \otimes \tau_2}}{\cdot} &=& \{(\W,
  %   (\fstore_1\uplus \fstore'_1, \tlang{(v_{1a},v_{2a})}),
  %   (\fstore_2\uplus \fstore'_2, \tlang{(v_{1b},v_{2b})}))
  %   \\ & & \quad \vo
  %   (\W, (\fstore_1, \tlang{v_{1a}}), (\fstore_2, \tlang{v_{1b}})) \in
  %   \Vrel{\alang{\tau_1}}{\cdot} \\ & & \quad \land
  %   (\W, (\fstore'_1, \tlang{v_{2a}}), (\fstore'_2, \tlang{v_{2b}})) \in \Vrel{\alang{\tau_2}}{\cdot}
  %   \}\\

  % \Vrel{\alang{\tau_1 \& \tau_2}}{\cdot} &=& \{(\W, (\fstore_1, (\tlang{\lambda
  %   \_.\{e_{1a}\}},\tlang{\lambda \_.\{e_{2a}\}})),\\
  % & & \qquad
  %     (\fstore_2, (\tlang{\lambda \_.\{e_{1b}\}},\tlang{\lambda \_.\{e_{2b}\}})))
  %   \\ & & \quad \vo
  %   (\W, (\fstore_1, \tlang{e_{1a}}), (\fstore_2, \tlang{e_{1b}})) \in
  %   \Erel{\alang{\tau_1}}{\cdot} \\ & & \quad \land
  %   (\W, (\fstore_1, \tlang{e_{2a}}), (\fstore_2, \tlang{e_{2b}})) \in \Erel{\alang{\tau_2}}{\cdot}
  %   \}
\end{array}
\]
\[
  \begin{array}{l}
    \Erel{\tau}{\rho} = \{(\W, (\fstore, \tlang{e})) \vo
    \fv(\tlang{e}) = \emptyset ~\land\\
    \qquad \forall \fstore_{r}, \heap {:} \W,~ \tlang{e'},~ \heap',~ j <
    \W.k.~ \fstore_{r}\#\fstore \land \fstore_{r}\uplus \fstore : \W\land \\
    \qquad \phconf{\fstore_{r}\uplus \rf(\W)\uplus \fstore}{\heap}{\tlang{e}} \phsteps{j} \phconf{\fstore'}{\heap'}{\tlang{e'}} \nrightarrow
    \\ \qquad\implies
    \tlang{\tlang{e'}} = \tlang{\convfail}
    \lor
    (\exists \fstore_{f}~\fstore_{g}~ \W'. \\
    \qquad\qquad \fstore'=\fstore_{r}\uplus \rf(\W')\uplus \fstore_{f}\uplus \fstore_{g} \\
    \qquad \qquad \land~\W \worldext_{\fstore_{r}} \W'
    \land~\heap' : \W' ~ \land~(\W', (\fstore_{f}, \tlang{e'})) \in \Vrel{\tau}{\rho})\}\\
  \end{array}
\]
     \[
      \begin{array}{l} (k,\heapty,\astore) \worldext_{\fstore}
(j,\heapty',\astore') \triangleq (j,\heapty',\astore')\in World_j ~\land~ \\
\qquad j \le k \land~\fstore \# \rf(k, \heapty,
\astore)~\land~\fstore \# \rf(j, \heapty', \astore')\\ \qquad \land~
\forall \loc \in \dom{\heapty}. \floor{j}{\heapty(\loc)} =
\heapty'(\loc) ~\land~ \\ \qquad \forall \loc \in
\dom{\astore}.(\loc) \in \dom{\astore'} \land\\ \qquad
\quad(\astore(\loc) = \aused \implies \astore'(\loc) = \aused)
\\ \qquad \quad ~\land~ (\astore(\loc) = \fstore \implies
\astore'(\loc) = (\aused \lor \fstore))
      \end{array}
    \]

\caption{Selections of $\langref$ \& $\langatlc$ Logical Relation.}
    \label{fig:affeff-lr}
  \end{small}
\end{figure}

With the logical relation in hand, we can prove analogous theorems to
Lemma~\ref{lemma:ref-conv-soundness} (Convertibility Soundness),
Theorem~\ref{theorem:ref-fundamental} (Fundamental Property),
Theorem~\ref{theorem:ref-safety-ll} (Type Safety for Lang A), and
Theorem~\ref{theorem:ref-safety-hl} (Type Safety for Lang B).

Note that to prove our type safety theorems, % Theorems \ref{theorem:affeff-langref-typesafe} and \ref{theorem:affeff-langatlc-typesafe}
we prove a lemma which states that, if $\tconf{\heap}{\tlang{e}}\steps{*}\tconf{\heap'}{\tlang{e'}}\nrightarrow$, then for any $\fstore$, $\phconf{\fstore}{\heap}{\tlang{e}}\phsteps{*}\phconf{\fstore_1'}{\heap_1'}{\tlang{e_1'}}\nrightarrow$. This lemma is necessary because the given assumption of the type safety theorem is that the configuration $\tconf{\heap}{\tlang{e}}$ steps under the normal operational semantics, but to apply the expression relation, we need that a corresponding configuration steps to an irreducible configuration under the phantom operational semantics.

Although our phantom flag realizability model was largely motivated by
efficiency concerns with the dynamic enforcement of affinity, more broadly, it
demonstrates how one can build complex static reasoning into the model even if
such reasoning is absent from the target. Indeed, the actual target language,
which source programs are compiled to and run in, has not changed; the
augmentations exist \emph{only in the model}. In this way, the preservation of
source invariants is subtle: it is not that the types actually exist in the
target (via runtime invariants or actual target types), but rather that the
operational behavior of the target is exactly what the type interpretations
characterize.

% Thus, realizability models like these
% allow us to move more of our reasoning into the model and out of the actual
% runtime that we deploy. And, from a very pragmatic angle, such models may allow
% one to reason post-hoc about compilers for existing systems that were not
% designed with static reasoning in mind.

%%% Local Variables:
%%% mode: latex
%%% TeX-master: "paper"
%%% End:

\section{Memory Management \& Polymorphism}\label{sec:polylin}

For our third case study, we consider how $\langsysf$, whose references are
garbage collected, can interoperate with core $\langlt$, a language with safe
strong updates despite memory aliasing, supported via linear
capabilities~\cite{ahmed07:L3}. This case study primarily highlights how
different memory management strategies can interoperate safely, in particular,
that manually managed linear references can be converted to garbage-collected
references without copying. This is of particular interest as more
low-level code is written in Rust, a language with an ownership discipline on
memory that similarly could allow safe transfer of memory to garbage-collected
languages.

% AHW: this para should probably be merged with the previous one
We also use this case study to explore how
polymorphism/generics in one language can be used, via a form of
interoperability, from the other. This is interesting because significant effort has gone into
adding generics to languages that did not originally support them, in order to
more easily build certain re-usable libraries.\footnote{e.g., Java 1.5/5, C\#
  2.0 \cite{kennedy01} and more recently, in the Go programming language} While
we are not claiming that interoperability could entirely replace built-in
polymorphism, sound support for cross-language type instantiation and
polymorphic libraries presents a possible alternative, especially for smaller,
perhaps more special-purpose, languages. This would allow us to write something
like:

\begin{small}
\[
  \flang{map \convert{\foreign{\llang{int}} \rightarrow \foreign{\llang{int}}}{\llang{(\lambda x:int. x+1)}} \convert{list~\foreign{\llang{int}}}{\llang{[1,2,3]}}}
\]
\end{small}

\noindent where the $\flang{blue}$ language supports polymorphism, and has a generic
$\flang{map}$ function, while the $\llang{pink}$ language does not. Of course,
since convertibility is still driving this, in addition to using a concrete
$\llang{intlist}$, $\llang{[1,2,3]}$, as above, the language without
polymorphism could convert entirely different (non-list) concrete
representations into similar polymorphic ones --- i.e., implementing a sort of
polymorphic interface at the boundary. For example, rather than an
$\llang{intlist}$ (or a $\llang{stringlist}$), in the example 
above, one could start with an $\llang{intarray}$ or $\llang{intbtree}$, or
any number of other traversable data structures that could be converted to
$\flang{list~int}$ (or any $\flang{list~\alpha}$).

\paragraph{Languages} We present the syntax of $\langlt$, augmented with forms
for interoperability, in Fig.~\ref{fig:polylin-syntax}. $\langlt$ has linear
capability types $\llang{\Capa{\zeta}{\tau}}$ (capability for abstract location
$\llang{\zeta}$ storing data of type $\llang{\tau}$), unrestricted pointer types
$\llang{\Ptr{\zeta}}$ to support aliasing, and location abstraction
($\llang{\Lambda} \llang{\zeta . e} : \llang{\forall \zeta . \tau}$ and
$\llang{\pack{\zeta,~v}} : \llang{\exists \zeta . \tau}$). The key insight to
$\langlt$ is that the pointer can be separated from the capability and passed
around in the program separately. At runtime, the capabilities will be erased,
but the static discipline only allows pointers to be used with their
capabilities (tied together with the type variables $\llang{\zeta}$), and only
allows capabilities to be used linearly. This enables safe in-place updates and
low-level manual memory management while still supporting some flexibility in
terms of pointer manipulation. We refer the reader to our supplementary
materials \cite{patterson22:semint-tr}, or the original paper on $\langlt$ (\cite{ahmed07:L3}) for more
details on its precise static semantics, but present highlights here. In
particular, $\llang{new}$ allocates memory and returns an existential package
containing a capability and pointer
($\llang{\exists \zeta. \Capa{\zeta}{\tau} \otimes \Ptr{\zeta}}$).
$\llang{swap}$ takes a matching capability ($\llang{\Capa{\zeta}{\tau_{1}}}$) and
pointer $\llang{\Ptr{\zeta}}$ and a value (of a possibly different type
$\llang{\tau_{2}}$) and replaces what is stored, returning the capability and
old value $\llang{\Capa{\zeta}{\tau_{2}} \otimes \tau_{1}}$. Note that since
capabilities record the type of what is in the heap and are unique,
strong updates are safe. Finally, $\llang{free}$ takes a package of a capability
and pointer ($\llang{\exists \zeta. \Capa{\zeta}{\tau} \otimes \Ptr{\zeta}}$)
and frees the memory, consuming both in the process and returning what was stored there---any lingering pointers are
harmless, as the necessary capability is now gone.

We compile both $\langlt$ and $\langsysf$ to an extension of the Scheme-like
target $\langtarget{}$ that we used in the previous case study (see
Fig.~\ref{fig:polylin-compilers} for $\langlt$; $\langsysf$ is standard). Our
additions to $\langtarget{}$, shown in Fig.~\ref{fig:polylin-target}, add manual memory allocation
($\tlang{alloc}$), $\tlang{free}$ (which will error on a garbage-collected
location), an instruction ($\gcmov$) to convert a manually managed location to
garbage collected, and an instruction ($\callgc$) to explicitly invoke the
garbage collector. The last allows the compiler to decide where the GC can
intercede (before allocation, in our compiler), and in doing so simplifies our
model slightly. The memory management itself is captured in our heap definition,
which allows the same location names to be used as either GC'd ($\gcmaps$) or
manually managed ($\manmaps$), and re-used after garbage collection or manual
free. Dereference ($\tlang{!e}$) and assignment 
($\tlang{e := e}$) work on both types of reference (failing, of course, if it is
manually managed and has been freed). This strategy of explicitly invoking the
garbage collector and using a single pool of locations retains significant
challenging aspects about garbage collectors while remaining simple enough to
expose the interesting aspects of interoperation.

As in the previous case study, we have boundary terms,
$\flang{\convert{\tau}{\llang{e}}}$ and $\llang{\convert{\tau}{\flang{e}}}$, for
\emph{converting} a term and using it in the other language. Now, we also add new types $\flang{\foreign{\llang{\tau}}}$, pronounced
``foreign type'', and allow conversions from $\llang{\tau}$ to
$\flang{\foreign{\llang{\tau}}}$ for \emph{opaquely embedding}\footnote{ Similar
  to ``lumps'' in Matthews-Findler\cite{matthews07}, though they give a
  \emph{single} lump type for all foreign types, i.e., they would have only
  $\flang{\foreign{}}$, rather than $\flang{\foreign{\llang{\tau}}}$. } types
for use in polymorphic functions.

If a language supports polymorphism, then its type abstractions should be
agnostic to the types that instantiate them, allowing them to range over not
only host types, but indeed any foreign types as well.  
Doing so should not violate parametricity. However, the non-polymorphic language
may need to make restrictions on how this power can be used, so as to not 
allow the polymorphic language to violate its invariants. To make this challenge
material, our non-polymorphic language in this case study has linear 
resources (heap capabilities) that cannot, if we are to maintain
soundness, be duplicated. This means, in particular, that whatever
interoperability strategy we come up with cannot allow a linear capability 
from $\langlt$ to flow over to a $\langsysf$ function that duplicates it, even
if such function is well-typed (and parametric) in $\langsysf$.
\begin{figure}
  \begin{small}
    \[
      \begin{array}{lcl}
        \langlt \\
        \text{Type~} \llang{\tau} & ::= & \llang{unit}
        \vo \llang{bool}
        \vo \llang{\tau \otimes \tau}
        \vo \llang{\tau \dynlolli \tau}
        \vo \llang{!\tau}
        \\ & &
        \vo \llang{\Ptr{\zeta}}
        \vo \llang{\Capa{\zeta}{\tau}}
        \vo \llang{\forall \zeta . \tau}
        \vo \llang{\exists \zeta . \tau} \\
        \text{Value~} \llang{v} &::=&
        \llang{\lambda x : \tau . e}
        \vo \llang{()}
        \vo \llang{\mathbb{B}}
        \vo \llang{(v, v)}
        \vo \llang{!v}
        \vo \llang{\Lambda} \llang{\zeta . e}
        \vo \llang{\pack{\zeta,~v}}
        \\
        \text{Expr.~} \llang{e} & ::= & \llang{v}
        \vo \llang{x}
        \vo \llang{(e,~e)}
        \vo \llang{e~e}
        \vo \llang{let~()~=~e~in~e}
        \vo \llang{if~e~e~e} \\ & &
        \vo \llang{let~(x,~x)~=~e~in~e}
        \vo \llang{let~!x~=~e~in~e}
        \vo \llang{dupl~e}  \\ & &
        \vo \llang{drop~e}
        \vo \llang{new~e} \vo \llang{free~e}
        \vo \llang{swap~e~e~e}
        \vo \llang{e~[\zeta]}  \\ & &
        \vo \llang{\pack{\zeta,~e}}
        \vo \llang{let~\pack{\zeta,~x}~=~e~in~e}
        \vo \llang{\convert{\tau}{\flang{e}}}
        \vo \llang{\foreign{\flang{e}}_{\tau}} \\

      \foreignset &= & \{\llang{unit}, \llang{bool}, \llang{\Ptr{\zeta}}, \llang{!\tau}\}
      \end{array}
    \]
    \vspace{-0.25cm}
    \caption{Syntax for $\langlt$.}
    \label{fig:polylin-syntax}
  \end{small}
\end{figure}

\begin{figure}
  \begin{small}
    \[
      \begin{array}{lcl}
        \text{Expr}~ \tlang{e} & ::= & \ldots \vo \tlang{alloc~e} \vo \tlang{free~e} \vo \tlang{\gcmov~e} \vo \callgc \\
        \text{Heap}~{\heap} & ::= & \loc \manmaps v, \heap \vo \loc \gcmaps v,
        \heap \vo \cdot \\
        \text{Err Code}~\tlang{c} & ::= & \ldots \vo \ptrerrcode \\
      \end{array}
    \]
    \vspace{-0.25cm}
    \caption{Additions to $\langtarget$ (see Fig.~\ref{fig:affeff-syntax} for base $\langtarget$).}
    \label{fig:polylin-target}
  \end{small}
\end{figure}

\begin{figure}
  \begin{small}
    % \[
    %   \begin{array}[t]{lcl}
    %     % sys f
    %     \flang{x},\llang{x} &\rightsquigarrow& \tlang{x} \\
    %     \flang{\lambda x : \tau . e} &\rightsquigarrow& \tlang{\lambda x . \flang{e^\tlang{+}}} \\
    %     \flang{e_1 e_2} &\rightsquigarrow& \flang{e_1^\tlang{+}} \flang{e_2^\tlang{+}} \\
    %   \end{array}
    %   \begin{array}[t]{lcl}
    %     \fcolor{\Lambda}\flang{\alpha . e} &\rightsquigarrow& \tlang{\lambda \_ . \flang{e^\tlang{+}}} \\
    %     \flang{e~[\tau]} &\rightsquigarrow& \tlang{\flang{e^\tlang{+}} ()} \\
    %     \flang{\convert{\tau}{\llang{e}}} &\rightsquigarrow&
    %     \tlang{C_{\llang{\tau} \mapsto \flang{\tau}}(\llang{e^\tlang{+}})}\\
    %   \end{array}
    % \]
    \[
      \llang{x} \rightsquigarrow \tlang{x} \hspace{0.2cm} \llang{()} \rightsquigarrow \tlang{()} \hspace{0.2cm} \llang{true/false} \rightsquigarrow \tlang{0/1} \hspace{0.2cm}  \llang{!v} \rightsquigarrow \llang{v}^+ \hspace{0.2cm} \llang{\lambda x : \tau . e} \rightsquigarrow \tlang{\lambda x.\llang{e}^+}
    \]
    \[
      \begin{array}[t]{lcl}
  \llang{e_1 e_2} &\rightsquigarrow& \llang{e_1}^+ \llang{e_2}^+ \\
  \llang{if~e_1~e_2~e_3} &\rightsquigarrow& \tlang{if~\llang{e_1}^+~\llang{e_2}^+~\llang{e_3}^+} \\
  \llang{(e_1,~e_2)} &\rightsquigarrow& \tlang{(\llang{e_1}^+,~\llang{e_2}^+)} \\
  \llang{dupl~e} &\rightsquigarrow& \tlang{let~x = \llang{e}^+~in~(x,~x)} \\
  \llang{drop~e} &\rightsquigarrow& \tlang{let~\_ = \llang{e}^+~in~()} \\
  \llang{new~e} &\rightsquigarrow& \tlang{let~\_ = callgc~in~let~x_\ell = alloc~\llang{e}^+}\\ & &\tlang{in~((),~x_\ell)} \\
        \llang{free~e} &\rightsquigarrow& \tlang{let~x = \llang{e}^+~in~let~x_r=~!(snd~x)~in}\\
         & & \tlang{\letc{\_}{free~(snd~x)}{x_r}} \\
  \llang{swap~e_c~e_p~e_v} &\rightsquigarrow&
\tlang{let~x_p = \llang{e_p}^+~in~
      let~\_ = \llang{e_c}~in~
        let~x_{v} = ~!x_p} \\ & & \tlang{in~
          \letc{\_}{(x_p := \llang{e_v}+)}{
                                  ((), x_{v})}} \\
  \llang{\Lambda}\llang{\zeta.e} &\rightsquigarrow& \tlang{\lambda \_.\llang{e}^+} \\
  \llang{e~[\zeta]} &\rightsquigarrow& \tlang{\llang{e}^+~()} \\
  \llang{\pack{\zeta,~e}} &\rightsquigarrow& \llang{e}^+ \\
  \llang{\convert{\tau}{\flang{e}}} &\rightsquigarrow& \tlang{C_{\flang{\tau} \mapsto \llang{\tau}}(\flang{e}^+)}
      \end{array}
    \]
    \[
    \begin{array}[t]{lcl}
\llang{let~()~=e_1~in~e_2} &\rightsquigarrow& \tlang{let~\_=\llang{e_1}^+~in~\llang{e_2}^+} \\
  \llang{let~(x_1,~x_2)~=e_1~in~e_2} &\rightsquigarrow& \tlang{let~p=\llang{e_1}^+~in~let~x_1=fst~p~in~}\\ & & \tlang{let~x_2=snd~p~in~\llang{e_2}^+} \\
  \llang{let~!x = e_1~in~e_2} &\rightsquigarrow& \tlang{let~x = \llang{e_1}^+~in~\llang{e_2}^+} \\
      \llang{\letc{\pack{\zeta,~x}}{e_1}{e_2}} &\rightsquigarrow&
    \tlang{let~x = \llang{e_1}^+~in~
          \llang{e_2}^+}\\
    \end{array}
    \]
    \vspace{-0.25cm}
    \caption{Compiler for $\langlt$.}
    \label{fig:polylin-compilers}
  \end{small}
\end{figure}

\paragraph{Convertibility}
The first conversion that we want to highlight is between references. In
$\langlt$, pointers have capabilities that convey ownership, and thus to convert
a pointer we also need the corresponding capability. For brevity, we may use $\llang{REF~\tau}$
to abbreviate a capability$+$pointer package type.

\begin{small}
\begin{mathpar}
  \inferrule{
    \tlang{\tlang{C_{\llang{\tau} \mapsto \flang{\tau}}}},
    \tlang{\tlang{C_{\flang{\tau} \mapsto \llang{\tau}}}}
    :
    \flang{\tau}
    \conv
    \llang{\tau}
  }{
    \tlang{\tlang{C_{\llang{REF~\tau} \mapsto \flang{ref~\tau}}}},
    \tlang{\tlang{C_{\flang{ref~\tau} \mapsto \llang{REF~\tau}}}}
    :
    \flang{ref~\tau}
    \conv
    \llang{\exists \zeta . \Capa{\zeta}{\tau}~\otimes~!\Ptr{\zeta}}
  }

  \begin{array}{l}
    \tlang{C_{\llang{REF~\tau}
        \mapsto \flang{ref~\tau}}}(\tlang{e})
    \triangleq
    \tlang{let~x=snd~e~in}\\
    \hspace{2.63cm}\tlang{let~\_=(x := C_{\llang{\tau} \mapsto
    \flang{\tau}}(!x))~in~\gcmov~x}\\
    \tlang{C_{\flang{ref~\tau} \mapsto \llang{REF~\tau}}}(\tlang{e})
    \triangleq
    \tlang{\letc{x}{alloc~C_{\flang{\tau} \mapsto
      \llang{\tau}}(!e)}{((),x)}}
    \\
  \end{array}
\end{mathpar}
\end{small}

The glue code itself is quite interesting: going from $\langlt$ to $\langsysf$,
since the $\langlt$ type system guarantees that this is the only capability to
this pointer, we can safely directly convert the pointer into a $\langsysf$
pointer with $\gcmov$ after in-place replacing the contents with the result of
converting (a less general rule that had a different premise might not need to
convert, e.g., if the data was already compatible---see the first case study for
more details). Going the other direction, from $\langsysf$ to $\langlt$, there
is no way for us to know if there are other aliases to the reference, so we
can't re-use the pointer. While we could simply disallow this conversion, and
error if it were attempted, instead we copy and convert data into a freshly
allocated manually managed location (note how, in the target, capabilities are erased
to unit). In this case, as in many, there are multiple sound ways of converting,
and it may be that a particular one makes more sense for your use case: we took
the position that it was useful to get a copy of the data, unaliased, but
perhaps a language designer would rather force the pointer to be dereferenced on the
$\langsysf$ side and the underlying data converted.

We account for interoperability of polymorphism in two parts. First, we have a
\emph{foreign type}, $\flang{\foreign{\llang{\tau}}}$, which embeds an $\langlt$
type into the type grammar of $\langsysf$. This foreign type, like any
$\langsysf$ type, can be used to instantiate type abstractions, define
functions, etc, but $\langsysf$ has no introduction or elimination rules for it---terms of foreign type must come across from, and then be sent back to, $\langlt$.
These come by way of the conversion rule $\flang{\foreign{\llang{\tau}}} \conv \llang{\tau}$,
which allow terms of the form $\flang{\convert{\foreign{\llang{\tau}}}{\llang{e}}}$
(to bring an $\langlt$ term to $\langsysf$) and
$\llang{\convert{\tau}{\flang{e}}}$ (the reverse). Moreover,
the conversion rule for foreign types restricts $\llang{\tau}$ to a safe
$\foreignset$ subset of types, but has no runtime consequences:

\begin{small}
\[
  \begin{array}{l}\inferrule{\llang{\tau}\in {\normalfont \foreignset}}{
    \tlang{C_{\flang{\foreign{\llang{\tau}}} \mapsto \llang{\tau}}},
    \tlang{C_{\llang{\tau} \mapsto \flang{\foreign{\llang{\tau}}}}}
    :
    \flang{\foreign{\llang{\tau}}}
    \conv
    \llang{\tau}
  }
\end{array}
\begin{array}{l}
  \tlang{C_{\flang{\foreign{\llang{\tau}}}{\mapsto}\llang{\tau}}(e) \triangleq e}\\
  \tlang{C_{\llang{\tau}{\mapsto}\flang{\foreign{\llang{\tau}}}}(e) \triangleq e}
\end{array}
\]
\end{small}

To prove soundness we need to show that $\foreignset$ types are indeed safe to embed.
The soundness condition depends on the 
expressive power of the two languages when viewed through the lens of
polymorphism. In our case, since the non-polymorphic language is linear but the polymorphic one is not, we need to show that
a $\foreignset$ type can be copied (i.e., none of its values own linear
capabilities)---this includes $\llang{unit}$ and $\llang{bool}$, but
also $\llang{\Ptr{\zeta}}$ and any type of the form $\llang{!\tau}$. Now, consider examples using this:

\begin{small}
  \begin{align}
  \fcolor{(\Lambda} \flang{\alpha.\lambda x{:}\alpha. \lambda y{:}\alpha.
      y)[\foreign{\llang{bool}}]~\flang{\convert{\foreign{\llang{bool}}}{\llang{true}}} ~\flang{\convert{\foreign{\llang{bool}}}{\llang{false}}}} \label{ex:polylin1} \\
  \flang{(\lambda x:BOOL. x) \convert{BOOL}{\llang{true}}} \text{~where~}    \flang{BOOL \triangleq \forall \alpha.\alpha \rightarrow \alpha \rightarrow
    \alpha} \label{ex:polylin2}
  \end{align}
\end{small}

In (\ref{ex:polylin1}), the leftmost expression is a polymorphic $\langsysf$
function that returns the second of its two arguments. It is instantiated it
with a foreign type, $\flang{\foreign{\llang{bool}}}$. Next, two terms of type
$\llang{bool}$ in $\langlt{}$ are embedded via the foreign conversion,
$\flang{\convert{\foreign{\llang{bool}}}{\llang{\cdot}}}$, which requires that
$\llang{bool} \in \foreignset$. Not only does this mechanism allow $\langlt{}$
programmers to use polymorphic functions, but also $\langsysf{}$ programmers to
use new base types. Of course, we could also convert the actual values, as in
(\ref{ex:polylin2}). To do so, we can define conversions between Church booleans
in $\langsysf$ (which has no booleans) and ordinary booleans in $\langlt$:

\begin{small}
\[\begin{array}{l}
    \inferrule{ }{\flang{\forall \alpha.\alpha \rightarrow \alpha \rightarrow
        \alpha} \conv \llang{bool}}
    \end{array}
\begin{array}{l}
    \tlang{C_{\flang{BOOL}{\mapsto}\llang{bool}}(e)} \triangleq \tlang{e~()~0~1}\\
    \tlang{C_{\llang{bool}{\mapsto}\flang{BOOL}}(e)} \triangleq \tlang{if0~e~\{}\fcolor{\Lambda} \flang{\alpha. \lambda x{:}\alpha. \lambda
  y{:}\alpha.x}\}~\\
  \hspace{3cm}\{\fcolor{\Lambda} \flang{\alpha. \lambda x{:}\alpha. \lambda
        y{:}\alpha.y}\}\\
  \end{array}
\]
\end{small}

\paragraph{Semantic Model}
In Fig.~\ref{fig:polylin-lr}, we present parts of the logical
relation that we use to prove our conversions and entire languages sound (see supplementary material \cite{patterson22:semint-tr}).

Our model is inspired by that of core $\langlt$ \cite{ahmed07:L3}, though ours
is significantly more complex to account for garbage collection
and interoperation with $\langsysf$. The key is a careful distinction between
owned (linear) manual memory, which is \emph{local} and described by heap
fragments associated with terms, and garbage-collected memory, which is
\emph{global} and described by the world $\W$. Since memory can be freed (via
garbage collection or manual $\tlang{free}$), reused, and moved from manual
memory to garbage-collected memory,
there are several constraints on how heap fragments and worlds may evolve so
we can ensure safe memory usage.

With that in mind, our value interpretation of source types $\Vrel{\tau}\rho$
are sets of worlds and related heap-fragments-and-values
$(\heap,\tlang{v})$, where the heap fragment
$\heap$ paired with value $\tlang{v}$ is the portion of the manually managed heap that $\tlang{v}$
\emph{owns}.

The relational substitution $\rho$ maps type variables $\flang{\alpha}$ to
arbitrary type interpretations $R$ and location variables $\llang{\zeta}$ to
concrete locations $\llang{\loc}$. Since $\langsysf$ cannot own
manual (linear) memory, all cases of $\Vrel{\flang{\tau}}\rho$ have empty
$\emptyset$ heap fragments. However, during evaluation, memory could be
allocated and subsequently freed so the expression relation does not have that
restriction. In $\langlt$, pointer types $\llang{ptr~\zeta}$ do not own
locations, so they can be freely copied. Rather, linear capabilities
$\llang{cap~\zeta~\tau}$ convey ownership of the location
$\tlang{\loc}$ that $\llang{\zeta}$ maps to and the heap
fragment $\tlang{H}$ pointed to by $\tlang{\loc}$.

In the expression relation $\Erel{\tau}{\rho}$, we run the expression with a set
of pinned locations ($\locsett$) that the garbage collector should not touch
(which may come from an outer context if we are evaluating a subterm), 
a garbage-collected heap fragment that satisfies the world ($\heap_{g+}$),
an arbitrary disjoint manually allocated ($MHeap$) ``rest'' of the heap
($\heap_{r}$), composed with the owned fragment ($\heap$). Then, assuming
$\tlang{e}$ terminates at $\tlang{v}$, we expect the ``rest'' heap is unchanged, the
garbage-collected portion has been transformed to $\heap_{g}'$, the owned
portion has been transformed into $\heap'$, and that
$(\W',(\heap',\tlang{v})) \in \Vrel{\llang{\tau}}\rho$,
where $\W'$ is a world the transformed GC'd portion of the heap $\heap_{g}'$ must satisfy.

Critical to the relation is world extension, written
$\worldext_{\locsets, \locbij}$, which indicates how our logical worlds can
evolve over time. In typical logical relations for state, the heap grows
monotonically and no location is ever overwritten, which world extension captures.
But, in our setting, the future heap might have deallocated, overwritten,
re-used memory (and re-used it between the GC and manual allocation). We can't
just allow arbitrary future states, however, as the semantics of types
do dictate restrictions on what has to happen in the heap. In particular, there
are two sets of locations that we need to keep careful track of: the rest can
change freely. The first are manually managed locations that we can't disturb, which index $\locsets$ captures.
Those are generally just the owned locations of term that we are currently
running. The second are the garbage collected locations that we must preserve in
the heap, at the same type (but we can change the value of), captured by $\locbij$. We also
have a syntactic shorthand, denoted by $\worldexte$, that
is indexed by the heap $\heap$ and the expressions $\tlang{e}$. This syntactic
shorthand is defined so that $\locsets$ takes its manually managed locations from the
domain of $\heap$ while $\locbij$ takes its garbage collected locations as the
locations in the original world that are present in either some value in the
heap $\heap$ or the expression $\tlang{e}$. Finally, we often use $\rchgclocs$ in
order to compute $\locbij$ when using world extension. $\rchgclocs(\W, S)$ is the set
of locations in the world $\W$ that are actually mentioned in the set $S$; i.e., $\rchgclocs(\W, S) = \dom{\W} \cap S$.

While our target supports dynamic failure (in the form of the $\tlang{fail}$
term), our logical relation rules out that possibility, ensuring that there are
no errors from the source nor from the conversion. This is, of course, a choice
we made, which may be stronger than desired for some languages (and, indeed, for
our previous two case studies), but given our choice of conversions, it is possible.

With the logical relation in hand, we prove analogous theorems to
Lemma~\ref{lemma:ref-conv-soundness} (Convertibility Soundness),
Theorem~\ref{theorem:ref-fundamental} (Fundamental Property),
Theorem~\ref{theorem:ref-safety-ll} (Type Safety for Lang A), and
Theorem~\ref{theorem:ref-safety-hl} (Type Safety for Lang B).

% \begin{lemma}[Convertibility Soundness]~

%   If $\flang{\tau} \conv \llang{\tau}$, then
%   $\begin{array}[t]{l}\forall (\W,(\heap_1,\tlang{e_1}),(\heap_2,\tlang{e_2})) \in \Erel{\flang{\tau}}{\rho}.~
%   (\W,(\heap_1,C_{\flang{\tau}{\mapsto}\llang{\tau}}(\tlang{e_1})),(\heap_2,C_{\flang{\tau}{\mapsto}\llang{\tau}}(\tlang{e_2}))) \in \Erel{\llang{\tau}}{\rho}\\
%   \forall (\W,(\heap_1,\tlang{e_1}),(\heap_2,\tlang{e_2})) \in \Erel{\llang{\tau}}{\rho}.~
%   (\W,(\heap_1,C_{\llang{\tau}{\mapsto}\flang{\tau}}(\tlang{e_1})),(\heap_2,C_{\llang{\tau}{\mapsto}\flang{\tau}}(\tlang{e_2}))) \in \Erel{\flang{\tau}}{\rho}\end{array}$
% \end{lemma}
% \begin{proof}
%   By induction on the convertibility relation. See supplementary material.
% \end{proof}

Our convertibility soundness result proves that our conversions above between
garbage-collected and manual references, as well as $\langlt$ booleans
and $\langsysf$ Church booleans (described above) are sound. We also show that
$\flang{\tau_1 \rightarrow \tau_2} \conv \llang{!(!\tau_1 \multimap \tau_2)}$
assuming $\flang{\tau_1} \conv \llang{\tau_1}$ and $\flang{\tau_2} \conv
\llang{\tau_2}$.

% \begin{theorem}[Fundamental Property]\label{theorem:polylin-fundamental}~

%   If $~\llang{\Delta}; \llang{\Gamma}; \fcolor{\Delta}; \fcolor{\Gamma}  \vdash
%   \flang{e} : \flang{\tau}$ then $\llang{\Delta}; \llang{\Gamma}; \fcolor{\Delta}; \fcolor{\Gamma}  \vdash
%   \flang{e} \aprx \flang{e} : \flang{\tau}$ and if $~\fcolor{\Delta}; \fcolor{\Gamma}; \llang{\Delta}; \llang{\Gamma}  \vdash
%   \llang{e} : \llang{\tau}$ then $\fcolor{\Delta}; \fcolor{\Gamma}; \llang{\Delta}; \llang{\Gamma}  \vdash
%   \llang{e} \aprx \llang{e} : \llang{\tau}$.

% \end{theorem}

% \begin{theorem}[Type Safety for $\langsysf$]{~}
% If $\llang{\cdot} ; \llang{\cdot} ; \fcolor{\cdot} ; \fcolor{\cdot}
%   \vdash \flang{e} : \flang{\tau}$, then
%   for any heap $\heap$, if $(\heap, \flang{e}^+)\steps{*} (\heap', \tlang{e'})$,
%   either there exist $\heap'', \tlang{e''}$ such that $(\heap', \tlang{e'})\step (\heap'', \tlang{e''})$
%   or $\tlang{e'}$ is a value.
% \end{theorem}

% \begin{theorem}[Type Safety for $\langlt$]{~}
% If $\fcolor{\cdot} ; \fcolor{\cdot} ;\llang{\cdot} ; \llang{\cdot}
%   \vdash \llang{e} : \llang{\tau}$, then
%   for any heap $\heap$, if $(\heap, \llang{e}^+)\steps{*} (\heap', \tlang{e'})$,
%   either there exist $\heap'', \tlang{e''}$ such that $(\heap', \tlang{e'})\step (\heap'', \tlang{e''})$
%   or $\tlang{e'}$ is a value.
% \end{theorem}

\begin{figure}
  \begin{small}
    \[\arraycolsep=1pt
      \begin{array}[t]{lcl}
          \Vrel{\flang{\alpha}}{\rho} &=& \atf{\rho}(\flang{\alpha}) \\
  \Vrel{\flang{unit}}{\rho} &=& \{\left(\W, (\emptyset, \tlang{()})\right)\} \\
  \Vrel{\flang{\tau_1 \rightarrow \tau_2}}{\rho} &=&
    \{(\W, (\emptyset, \tlang{\lambda x . e})) \vo \forall \W', \tlang{v}.~ \W\worldexte_{\emptyset, \tlang{e}} \W' \land \\ & & \hspace{-1.25cm}(\W', (\emptyset, \tlang{v})) \in \Vrel{\flang{\tau_1}}{\rho} \implies
      (\W', (\emptyset, \tlang{[x \mapsto \tlang{v}]e})) \in \Erel{\flang{\tau_2}}{\rho}
    \} \\
  \Vrel{\flang{\forall \alpha . \tau}}{\rho} &=&
    \{(\W, (\emptyset, \tlang{\lambda\_.e}),) \vo \forall R \in RelT, \W'.  \\ & & \hspace{-1.25cm}
        \W\worldextstricte_{\emptyset, \tlang{e}} \W'\implies (\W', (\emptyset, \tlang{e})) \in \Erel{\flang{\tau}}{\rho[\flang{F}(\flang{\alpha}) \mapsto R]} \} \\
  \Vrel{\flang{ref~\tau}}{\rho} &=&
    \{(\W, (\emptyset, \tlang{\ell})) \vo \W.\heapty(\loc) = \floor{\W.k}{\Vrel{\rlang{\tau}}{\rho}} \} \\
  \Vrel{\flang{\foreign{\llang{\tau}}}}{\rho} &=& \Vrel{\llang{\tau}}{\rho} \\
  \Vrel{\llang{unit}}{\rho} &=& \{\left(\W, (\emptyset, \tlang{()}) \right)\} \\
  \Vrel{\llang{bool}}{\rho} &=&
    \left\{(\W, (\emptyset, \tlang{b})) \vo \tlang{b} \in \left\{\tlang{0}, \tlang{1}\right\}\right\} \\
  \Vrel{\llang{\tau_1~\otimes~\tau_2}}{\rho} &=&
    \{(\W, (\heap_{1} \uplus \heap_{2}, \tlang{(\tlang{v_{1}}},~\tlang{v_{2}}))) \vo \\
      & & \hspace{-0.5cm}(\W, (\heap_{1}, \tlang{v_{1}})) \in \Vrel{\llang{\tau_1}}{\rho} \land (\W, (\heap_{2}, \tlang{v_{2}})) \in \Vrel{\llang{\tau_2}}{\rho}\} \\
  \Vrel{\llang{\tau_1 \dynlolli \tau_2}}{\rho} &=&
    \{(\W, (\heap, \tlang{\lambda x . e})) \vo
      \forall \W', \heap_{\tlang{v}}, \tlang{v}. \\
      & & \hspace{-1.5cm} \W\worldexte_{\heap, \tlang{e}} \W' \land(\W', (\heap_{\tlang{v}}, \tlang{v})) \in \Vrel{\llang{\tau_1}}{\rho} \implies\\
      & & \hspace{-2cm}  (\W', (\heap \uplus \heap_{\tlang{v}}, \tlang{[x \mapsto \tlang{v}]e})) \in \Erel{\llang{\tau_2}}{\rho}\} \\
  \Vrel{\llang{!\tau}}{\rho} &=&
    \{(\W, (\emptyset, \tlang{v})) \vo (\W, (\emptyset, \tlang{v})) \in \Vrel{\llang{\tau}}{\rho} \} \\
  \Vrel{\llang{\Ptr{\zeta}}}{\rho} &=&
    \{(\W, (\emptyset, \tlang{\ell})) \vo
      \atl{\rho}(\llang{\zeta}) = \tlang{\ell} \} \\
  \Vrel{\llang{\Capa{\zeta}{\tau}}}{\rho} &=&
    \{(\W, (\heap \uplus \{\tlang{\ell} \mapsto \tlang{v}\}, ())) \vo \\ & & \hspace{-.5cm}
      \atl{\rho}(\llang{\zeta}) = \tlang{\ell} \land (\W, (\heap, \tlang{v})) \in \Vrel{\llang{\tau}}{\rho} \} \\
  \Vrel{\llang{\forall \zeta . \tau}}{\rho} &=&
    \{(\W, (\heap, \tlang{\lambda \_ . e})) \vo \\ & & \quad \forall \loc.\, (\W, (\heap, \tlang{e})) \in \Erel{\llang{\tau}}{\rho[\llang{L3}(\llang{\zeta}) \mapsto \tlang{\ell}]} \} \\
  \Vrel{\llang{\exists \zeta . \tau}}{\rho} &=&
                                                \{(\W, (\heap, \tlang{ \tlang{v}})) \vo
            \exists \loc.\, (\W, (\heap, \tlang{v})) \in \Vrel{\llang{\tau}}{\rho[\llang{L3}(\llang{\zeta}) \mapsto \tlang{\ell}]} \}
    \end{array}
    \]
    \[\arraycolsep=1pt
      \begin{array}{l}
  \Erel{\tau}{\rho} =
    \{ (\W, (\heap, \tlang{e})) \vo
      \forall \locsett, \tlang{v}, \heap_{g+} : \W, \mpart{\heap_{r}}, \heap_{*} . \\ \quad
        (\heap_{g+}\uplus \heap \uplus \heap_{r}, \tlang{e})
        \multistep_{\locsett}
        (\heap_{*}, \tlang{v}) \nrightarrow_{\locsett} \\ \quad
        \implies \exists \heap', \heap_{g}'. \exists \W'.
          \heap_{*} = \heap_{g}' \uplus \heap'\uplus \heap_{r}  ~\land \heap_{g}' : \W' ~\land \\
          \W \worldext_{(\dom{\heap_{r}}), \rchgclocs(\W, \locsett\cup \fl(\cod(\heap_{r})))} \W' \\ \qquad
          \land ~(\W', (\heap', \tlang{v})) \in \Vrel{\tau}{\rho}
          \fcolor{~\land~ \heap_{1'}=\emptyset}\}
\end{array}
\]

\[
(k,\heapty) \worldext_{\locsets, \locbij} (j,\heapty') = \begin{array}[t]{l}
  j \le k
  \land~ \locsets\# \dom{\heapty'} \\
  \land~\forall \loc \in \locbij. \heapty'(\loc) = \floor{j}{\heapty(\loc)}
\end{array}
\]

Note the $\fcolor{\text{highlighted parts}}$ only apply to $\langsysf$ types.
          % \[\arraycolsep=1pt
          % \begin{array}{lcl}
          %   \lts(\rho) &\equiv & \{\tlang{x_\zeta}\mapsto (\tlang{\ell_1}, \tlang{\ell_2}) \mid \llang{\zeta}\mapsto (\tlang{\ell_1}, \tlang{\ell_2})\in \rho\}
          % \end{array}
          % \]
%               \[
%               \begin{array}{l}
%   \llang{\Delta} ; \llang{\Gamma} ; \fcolor{\Delta} ; \fcolor{\Gamma}
%   \vdash \flang{e_1} \aprx \flang{e_2} : \flang{\tau} \equiv\\
%   \quad \forall \rho, \lsub, \fsub.
%     \atl{\rho} \in \Tyenvrel{\llang{\Delta}} \land
%     \atf{\rho} \in \Tyenvrel{\fcolor{\Delta}} \land
%     (\emptyset, \emptyset, \lgammasub) \in \Envrel{\llang{\Gamma}}{\rho} \land
%     \fsub \in \Envrel{\fcolor{\Gamma}}{\rho}
%     \land~\ldeltasub=\lts(\atl{\rho})  \\ \qquad\qquad
%     \implies (\emptyset,
%               \lsub^1(\fsub^1(\flang{e_1}^+)),
%               \emptyset,
%               \lsub^2(\fsub^2(\flang{e_2}^+)))
%             \in \Erel{\flang{\tau}}{\rho} \\
%   \fcolor{\Delta} ; \fcolor{\Gamma} ;\llang{\Delta} ; \llang{\Gamma}
%   \vdash \llang{e_1} \aprx \llang{e_2} : \llang{\tau} \equiv\\ \quad
%   \forall \rho, \fsub, \lsub, \heap_1, \heap_2 .
%     \atf{\rho} \in \Tyenvrel{\fcolor{\Delta}} \land
%     \atl{\rho} \in \Tyenvrel{\llang{\Delta}} \land
%     \fsub \in \Envrel{\fcolor{\Gamma}}{\rho} \land
%     (\heap_1, \heap_2, \lgammasub) \in \Envrel{\llang{\Gamma}}{\rho} \land~\ldeltasub=\lts(\atl{\rho})
%     \\  \qquad\qquad
%     \implies (\heap_1,
%               \fsub^1(\lsub^1(\llang{e_1}^+)),
%               \heap_2,
%               \fsub^2(\lsub^2(\llang{e_2}^+)))
%             \in \Erel{\llang{\tau}}{\rho}
% \end{array}
% \]
    \caption{Logical Relation for $\langsysf$ and $\langlt$.}
    \label{fig:polylin-lr}
  \end{small}
\end{figure}

\paragraph{Discussion} While we showed how to handle universal types, handling
existential types is another question. With our existing ``foreign type''
mechanism, we can support defining data structures and operations over them and
passing both. For example, we could pass an expression of type $\foreign{int}
\times \foreign{int} \rightarrow \foreign{int} \times \foreign{int} \rightarrow
int$, for a counter defined as an integer. That provides some degree of
abstraction, but doesn't, for example, disallow passing the $\foreign{int}$ back
to some other code that expects that type. We could, however, in the language
with existential types, pack that to $\exists \alpha. \alpha \times \alpha
\rightarrow \alpha \times \alpha \rightarrow int$.

More interesting is the question when both languages have polymorphism. In that
case, if we wanted to convert abstract types, we would need to generalize our
convertibility rules to handle open types, i.e.,
$\Delta \vdash \tau \conv \tau'$. If the interpretation of type variables were
the same in both languages (i.e., in our model this would mean that both were
drawn from the same relation), this would be sufficient. If, however, the
interpretation of type variables were different in the two languages (we do this
in the case study in \S\ref{sec:affeff}, see our supplemental materials \cite{patterson22:semint-tr} for the use of $\UnrTyp$ in $\Vrel{\rlang{\forall \alpha . \tau}}\rho$), we
would need, in our source type systems, some form of bounded polymorphism in
order to restrict the judgment to variables that were equivalent. Otherwise, it
would be impossible to prove convertibility rules sound.

%%% Local Variables:
%%% mode: latex
%%% TeX-master: "paper"
%%% End:

\section{Related Work and Conclusion}\label{sec:related}

Most research on interoperability has focused either on reducing
boilerplate or improving performance. We will not discuss those, focusing on work
addressing soundness.

\para{Multi-language semantics.} \citet{matthews07} studied the question of the
interoperability of source languages, developing the idea of a syntactic
multi-language with \emph{boundary terms} (c.f., contracts
\cite{findler2002contracts,findler2006contracts}) that mediate between the two
languages. They focused on a static language interacting with a dynamic one, but
similar techniques have been applied widely (e.g., object-oriented
\cite{gray2005fine,gray2008safe}, %technically not MF07!
affine and unrestricted \cite{tov10}, simple and dependently typed
\cite{osera12}, functional language and assembly \cite{patterson17}, linear and
unrestricted \cite{scherer18}) and used to prove compiler properties (e.g.,
correctness \cite{perconti14}, full abstraction \cite{ahmed11,new16}). More
recently, there has been an effort understand this construction from a
denotational \cite{buro2019multi} and categorical \cite{buro2020equational}
perspective. While the last may seem particularly relevant to our work, they
still firmly root the multi-language as a source-language construct, rather than
building it out of a common substrate, our key divergence from this prior work.

\citet{barrett16} take a slightly different path, directly mixing languages (PHP and Python) and allowing bindings from one to be used in the
other, though to similar ends.

\para{Interoperability via typed targets.} Shao and Trifonov
\cite{shao98,trifonov99} studied interoperability much earlier, and closer to
our context: they consider interoperability mediated by translation to a common
target. They tackle the problem that one language has access to control effects
and the other does not. Their approach, however, is different: it relies upon a
target language with an effect-based type system that is sufficient to capture
the safety invariants, whereas while our realizability approach can certainly
benefit from typed target languages, it doesn't rely upon them. While typed
intermediate languages obviously offer real benefits, there are also unaddressed
problems, foremost of which is designing a usable type system that is
sufficiently general to allow (efficient) compilation from all the languages you
want to support. While there are ongoing attempts (probably foremost is the
TruffleVM project \cite{grimmer15}) to design such general intermediates, most
have focused their attention on untyped or unsound languages, and in the
particular case of TruffleVM, there is as-yet no meta-theory.

\para{An abstract framework for unsafe FFIs.} \citet{turcotte19} advocate a
framework using an abstract version of the foreign language, so
soundness can be proved without building a full multi-language. They demonstrate
this by proving a modified type safety proof of Lua and C interacting via the C
FFI, modeling the C as code that can do arbitrary unsound behavior and thus
blamed for all unsoundness. While this approach seems promising in the
context of unsound languages, it is less clear
how it applies to sound languages.

% \para{Modeling FFIs via State Machines.} \citet{lee10} specify the type (and
% other) constraints that exist in both the JNI and Python/C FFI via state
% machines and use that to generate runtime checks to enforce these at runtime.
% While this is practical work and so they do not prove properties about their
% system, Jinn, there are many similarities between their approach and ours. In
% particular, the idea that invariants that cannot be expressed via the languages
% themselves and should instead be checked via inserted code. We would expect that if their
% approach were applied to safe languages, we would be able to prove that the
% code that they inserted satisfied semantic interpretations of the respective
% types.

% \paragraph{Logical relations techniques}
% In terms of semantic techniques, our step-indexed model for state builds upon work by
% \citet{ahmed04:phd} and \citet{appel01}, and is inspired by, though simpler
% than, later work by \citet{ahmed09} and \citet{dreyer10}. Our
% realizability-style model is inspired by work by
% \citet{benton09:tldi,benton2009}: in their case their relations are used to
% prove compiler correctness, but their technique is similar.

\para{Semantic Models and Realizability Models} The use of semantic models to
prove type soundness has a long history~\cite{milner78}.  We make use of
step-indexed models~\cite{appel01,ahmed04:phd}, developed as part of the
Foundational Proof-Carrying Code~\cite{ahmed10:fpcc-jrnl} project, which showed
how to scale the semantic approach to complex features found in
real languages such as recursive types and higher-order mutable state.  While
much of the recent work that uses step-indexed models is concerned with program
equivalence, one recent project that focuses on type soundness is 
RustBelt~\cite{jung18:rustbelt}: they give a semantic model of $\lambda_{Rust}$
types and use it to prove the soundness of $\lambda_{Rust}$ typing rules, but
also to prove that the $\lambda_{Rust}$ implementation of standard library
features (essentially unsafe code) are semantically sound inhabitants of their
ascribed type specification. 

Unlike the above, our realizability model interprets source types as sets of
target terms. Our work takes inspiration from a line of work by Benton and
collaborators on ``low-level semantics for high-level types'' (dubbed
``realistic realizability'')~\cite{benton06:new}. Such models were used to prove
type soundness of standalone languages, specifically, \citet{benton07:ppdp}
proved an imperative while language sound and \citet{benton09:tldi} proved type
soundness for a simply typed functional language, both times interpreting source
types as relations on terms of an idealized assembly and allowing for compiled
code to be linked with a verified memory allocation module implemented in
assembly~\cite{benton06:new}. \citet{krishnaswami15} make use of a realizability
model to prove consistency of $\mathrm{LNL}_{D}$ a core type theory that
integrates linearity and full type dependency. The linear parts of their model,
like our interpretation of $\langlt$ types, are directly inspired by the
semantic model for $\langlt$ by \citet{ahmed07:L3}. While they consider
interoperability and use realizability models, their approach is quite different
from ours, as their introduce both term constructors and types ($G$
and $F$) that allow direct embedding into the other language, thereby
changing it, rather than defining conversions into existing types (which,
indeed, is probably impossible in their case). More generally, such
realizability models have also been used by \citet{jensen13} to verify low-level
code using a high-level separation logic, and by \citet{Benton2009} to verify
compiler correctness.

Finally, New et al.~\cite{newahmed18,new2019,new20:gradparam} make use of
realizability models in their work on 
semantic foundations of gradual typing, work that we have drawn inspiration from,
given gradual typing is a special instance of language
interoperability. They compile type casts in a surface gradual language to a
target Call-By-Push-Value \cite{levy01:phd} language without casts, build a realizability model of
gradual types and type precision as relations on target terms, and prove
properties about the gradual surface language using the model.

\para{Verification-based Approaches} Much work has been done using high-level
program logics to reason about target terms, which can be seen as analogous to
the realizability approach. Perhaps most relevant, in the context of
interoperability, is the Cito system of \citet{wang14}, where code to-be-linked
is given a specification over the behavior of target code, and compilation can
then proceed relying upon that specification. This clearly renders benefits in
terms of language independence, since any compiled code that satisfied that
specification could be used. However, there is a significant difference from our
work: by incorporating the semantics of types of both languages we can prove
that the \emph{conversions} preserve those semantics, and thus allow an end user
to gain the benefits of type soundness without having to do any verification.
Indeed, proving the conversions sound (or, in the case that they can be no-ops,
proving that is okay) is the central result of this paper, and such conversions
are not a part of the setup of \citet{wang14}.

%%% Local Variables:
%%% mode: latex
%%% TeX-master: "paper"
%%% End:

%\vspace{-0.5em}
\paragraph{Conclusion and Future Work}\label{sec:conclusion}

We have presented a novel framework for the design and
verification of sound language interoperability where that interoperability
happens, as in practical systems, after compilation. 
The realizability models at the heart of our technique give us powerful
reasoning tools, including the ability to encode static invariants that are
otherwise impossible to express in often untyped or low-level target languages.
Even when it is possible to turn static source-level invariants into dynamic
target-level checks, the ability to instead move these invariants into the model
allows for more performant (and perhaps, realistic) compilers without losing the
ability to prove soundness.

In the future, we hope to apply the framework to further explorations
of the interoperability design space, e.g., to investigate interactions between
lazy and strict languages (compilation to Call-By-Push-Value \cite{levy01:phd}
may illuminate conversions), between single-threaded and concurrent languages  
(session types \cite{honda1993types,takeuchi1994interaction,honda1998language} 
may help guide interoperability with process calculi like the $\pi$-calculus
\cite{milner1992calculus}), between different control effects, and between
Rust and a GC'ed language such as ML, Java, or Haskell compiled to a low-level target. 
% about control effects, not to mention further explorations of polymorphism as
% described at the end of \S\ref{sec:polylin}. 

%\clearpage

%Acknowledgments
\begin{acks}                            %% acks environment is optional
                                        %% contents suppressed with 'anonymous'
  %% Commands \grantsponsor{<sponsorID>}{<name>}{<url>} and
  %% \grantnum[<url>]{<sponsorID>}{<number>} should be used to
  %% acknowledge financial support and will be used by metadata
  %% extraction tools.
  We thank the anonymous reviewers for their in-depth comments.
  This material is based upon work supported by the \grantsponsor{GS100000001}{National Science Foundation}{} under Grant No.~\grantnum{GS100000001}{CCF-1816837} and \grantnum{GS100000001}{CCF-1453796}.
  %and do not necessarily reflect the views of the
  %National Science Foundation.
\end{acks}

\clearpage

%% Bibliography
\bibliography{dbp}

%% Appendix
% \appendix
% \section{Appendix}

% Text of appendix \ldots

\end{document}